\documentclass[onecolumn,a4paper,accepted=2024-04-29]{quantumarticle}
\pdfoutput=1

\usepackage{authblk}

\title{Resource Efficient Boolean Function Solver on Quantum Computer}
\author[1]{Xiang Li}
\author[2]{Hanxiang Shen}
\author[1,3]{Weiguo Gao}
\author[1,3]{Yingzhou Li}

\affil[1]{School of Mathematical Sciences, Fudan University}
\affil[2]{Shanghai Center for Mathematical Science, Fudan University}
\affil[3]{Shanghai Key Laboratory for Contemporary Applied Mathematics}

\usepackage{graphicx}
\graphicspath{{figure/}}
\usepackage{subcaption}

\usepackage{xcolor}
\usepackage[colorlinks,anchorcolor=blue]{hyperref}
\usepackage{amsmath, amssymb, amsthm}
\newtheorem{theorem}{Theorem}
\newtheorem{corollary}{Corollary}

\usepackage{algpseudocode}
\usepackage[Algorithmus]{algorithm}

\algnewcommand{\IIf}[1]{\State\algorithmicif\ #1\ \algorithmicthen}
\algnewcommand{\EndIIf}{\unskip\ \algorithmicend\ \algorithmicif}

\usepackage{bm}
\usepackage{cite}
\usepackage{extarrows}
\usepackage{booktabs}
\usepackage{braket}
\usepackage{qcircuit}
\usepackage{diagbox}
\usepackage{tikz}

\usepackage{multirow}
\usepackage{multicol}

\newcommand{\op}{\texttt}

\newcommand{\bigO}[1]{\mathcal{O}\left(#1\right)}

\newcommand{\U}[2]{U^{(#1)}_{#2}}
\newcommand{\N}[2]{N^{(#1)}_{#2}}
\newcommand{\T}[2]{T^{(#1)}_{#2}}
\newcommand{\F}[2]{F^{(#1)}_{#2}}
\newcommand{\G}[2]{G^{(#1)}_{#2}}
\newcommand{\K}[2]{K^{(#1)}_{#2}}

\newcommand{\bbC}{\mathbb{C}}
\newcommand{\bbE}{\mathbb{E}}
\newcommand{\bbF}{\mathbb{F}}

\newcommand{\calF}{\mathcal{F}}
\newcommand{\calG}{\mathcal{G}}
\newcommand{\calQ}{\mathcal{Q}}
\newcommand{\calR}{\mathcal{R}}
\newcommand{\calS}{\mathcal{S}}

\newcommand{\qubit}[1]{\ensuremath{\ket{#1}}}
\newcommand{\abs}[1]{\ensuremath{\lvert#1\rvert}}

\newcommand{\qcExampleNOTGates}{
    \Qcircuit @C=1.5em @R=.7em {
        \lstick{\qubit{x_1}} & \qw & \qw\barrier{2} & \qw & \gate{\op{X}} & \qw\barrier{2} & \qw & \ctrl{2} & \qw\barrier{2} & \qw & \ctrl{1} & \qw\barrier{2} & \qw & \gate{\op{Z}} & \qw\barrier{2} & \qw & \ctrl{1} & \qw\barrier{2} & \qw \\
        \lstick{\qubit{x_2}} & \qw & \qw & \qw & \qw & \qw & \qw & \qw & \qw & \qw & \ctrl{1} & \qw & \qw & \qw & \qw & \qw & \ctrl{1} & \qw & \qw \\
        \lstick{\qubit{x_3}} & \qw & \qw & \qw & \qw & \qw & \qw & \targ & \qw & \qw & \targ & \qw & \qw & \qw & \qw & \qw & \ctrl{0} & \qw & \qw \\
         &  &  &  & \dstick{\op{NOT}} &  &  & \dstick{\op{CNOT}} &  &  & \dstick{\op{CCNOT}} &  &  & \dstick{\op{Z}} &  &  & \dstick{\op{MCZ}} &  &  \\
         &  &  &  &  &  &  &  &  &  &  &  &  &  &  &  &  &  &  \\
         &  & \dstick{\qubit{000}} &  &  & \dstick{\qubit{100}} &  &  & \dstick{\qubit{101}} &  &  & \dstick{\qubit{101}} &  &  & \dstick{-\qubit{101}} &  &  & \dstick{-\qubit{101}} &  \\
         &  &  &  &  &  &  &  &  &  &  &  &  &  &  &  &  &  &
    }
}

\newcommand{\qcExampleFunctionControl}{
    \Qcircuit @C=1.5em @R=.7em {
        \lstick{\qubit{x_1}} & \ctrl{2} & \ctrl{1} & \qw      & \qw
        & \ctrl{2} & \ctrl{1} & \qw      & \qw
        & & & & & &\\
        \lstick{\qubit{x_2}} & \qw      & \ctrl{1} & \qw      & \qw
        & \qw      & \ctrl{1} & \qw      & \qw
        & \xlongrightarrow{abbr.} & &
        \lstick{\qubit{x}} & /_2\qw & \sgate{f}{1} & \qw
        & \sgate{f}{1} & \qw \\
        \lstick{\qubit{0}}   & \targ    & \targ    & \gate{\op{X}} & \gate{\op{Z}}
        & \targ    & \targ    & \gate{\op{X}} & \qw
        & & &
        \lstick{\qubit{0}} & \qw    & \targ        & \gate{\op{Z}}
        & \targ        & \qw \gategroup{1}{2}{3}{4}{2em}{--} \\
    }
}

\newcommand{\qcExampleStackAll}{
    \Qcircuit @C=.5em @R=.7em {
        \lstick{\qubit{x}} & {/} \qw & \qw & \qw & \sgate{f_3}{3} & \sgate{f_2}{2} & \sgate{f_1}{1} & \qw & \qw & \qw & \sgate{f_1}{1} & \sgate{f_2}{2} & \sgate{f_3}{3} & \qw & \qw & \qw \\
        \lstick{\qubit{0}} & \qw & \qw & \qw & \qw & \qw & \targ & \qw & \ctrl{1} & \qw & \targ & \qw & \qw & \qw & \qw & \qw \\
        \lstick{\qubit{0}} & \qw & \qw & \qw & \qw & \targ & \qw & \qw & \ctrl{1} & \qw & \qw & \targ & \qw & \qw & \qw & \qw \\
        \lstick{\qubit{0}} & \qw & \qw & \qw & \targ & \qw & \qw & \qw & \ctrl{0} & \qw & \qw & \qw & \targ & \qw & \qw & \qw \\
         &  &  &  &  &  &  &  &  &  &  &  &  &  &  &
    }
}

\newcommand{\qcExampleRecursiveLevelOne}{
    \Qcircuit @C=.5em @R=.7em {
        \lstick{\qubit{x}} &{/} \qw & \qw & \qw & \sgate{f_1}{3} & \sgate{f_2}{2} & \sgate{f_3}{1} & \qw & \qw & \qw & \sgate{f_3}{1} & \sgate{f_2}{2} & \sgate{f_1}{3} & \qw & \qw & \qw \\
        \lstick{a_1}& \qw & \qw & \qw & \qw & \qw & \targ & \qw & \ctrl{1} & \qw & \targ & \qw & \qw &  \qw & \qw & \qw \\
        \lstick{a_2}& \qw & \qw & \qw & \qw & \targ & \qw & \qw & \ctrl{1} & \qw & \qw & \targ & \qw & \qw & \qw & \qw \\
        \lstick{a_3}& \qw & \qw & \qw & \targ & \qw & \qw & \qw & \ctrl{1} & \qw & \qw  & \qw & \targ &  \qw & \qw & \qw \\
        \lstick{a_4}& \qw & \qw & \qw & \qw & \qw & \qw & \qw & \targ & \qw & \qw & \qw & \qw & \qw & \qw & \qw \gategroup{1}{5}{6}{13}{2em}{--}\\
        & & & & &  & & & & &  & & {U_4^{(1)}} & & &
    }
}

\newcommand{\qcExampleRecursiveLevelK}{
    \Qcircuit @C=.5em @R=.7em {
        \lstick{\qubit{x}} &{/} \qw & \qw & \qw & \multigate{3}{U_3^{(\ell-1)}} & \multigate{2}{U_2^{(\ell-1)}} & \multigate{1}{U_1^{(\ell-1)}} & \qw & \qw & \qw & \multigate{1}{U_1^{(\ell-1)}} & \multigate{2}{U_2^{(\ell-1)}} & \multigate{3}{U_3^{(\ell-1)}} & \qw & \qw & \qw \\
        \lstick{a_1}& \qw & \qw & \qw & \ghost{U_3^{(\ell-1)}} & \ghost{U_2^{(\ell-1)}} & \ghost{U_1^{(\ell-1)}} & \qw & \ctrl{1} & \qw & \ghost{U_1^{(\ell-1)}} & \ghost{U_2^{(\ell-1)}} & \ghost{U_3^{(\ell-1)}} &  \qw & \qw & \qw \\
        \lstick{a_2}& \qw & \qw & \qw & \ghost{U_3^{(\ell-1)}} & \ghost{U_2^{(\ell-1)}} & \qw & \qw & \ctrl{1} & \qw & \qw & \ghost{U_2^{(\ell-1)}} & \ghost{U_3^{(\ell-1)}} & \qw & \qw & \qw \\
        \lstick{a_3}& \qw & \qw & \qw & \ghost{U_3^{(\ell-1)}} & \qw & \qw & \qw & \ctrl{1} & \qw & \qw  & \qw & \ghost{U_3^{(\ell-1)}} &  \qw & \qw & \qw \\
        \lstick{a_4}& \qw & \qw & \qw & \qw & \qw & \qw & \qw & \targ & \qw & \qw & \qw & \qw & \qw & \qw & \qw \gategroup{1}{5}{6}{13}{2em}{--}\\
        & & & & &  & & & & &  & & {U_4^{(\ell)}} & & &
    }
}

\newcommand{\qcExampleRecursiveCircuit}{
    \Qcircuit @C=.5em @R=.7em {
        \lstick{\qubit{x}} & {/} \qw & \qw & \qw & \multigate{4}{U_4^{(l-1)}} & \multigate{3}{U_3^{(l-1)}} & \multigate{2}{U_2^{(l-1)}} & \multigate{1}{U_1^{(l-1)}} & \qw & \qw & \qw & \multigate{1}{U_1^{(l-1)}} & \multigate{2}{U_2^{(l-1)}} & \multigate{3}{U_3^{(l-1)}} & \multigate{4}{U_4^{(l-1)}} & \qw & \qw & \qw \\
        \lstick{a_1} & \qw & \qw & \qw & \ghost{U_4^{(l-1)}} & \ghost{U_3^{(l-1)}} & \ghost{U_2^{(l-1)}} & \ghost{U_1^{(l-1)}} & \qw & \ctrl{1} & \qw & \ghost{U_1^{(l-1)}} & \ghost{U_2^{(l-1)}} & \ghost{U_3^{(l-1)}} & \ghost{U_4^{(l-1)}} & \qw & \qw & \qw \\
        \lstick{a_2} & \qw & \qw & \qw & \ghost{U_4^{(l-1)}} & \ghost{U_3^{(l-1)}} & \ghost{U_2^{(l-1)}} & \qw & \qw & \ctrl{1} & \qw & \qw & \ghost{U_2^{(l-1)}} & \ghost{U_3^{(l-1)}} & \ghost{U_4^{(l-1)}} & \qw & \qw & \qw \\
        \lstick{a_3} & \qw & \qw & \qw & \ghost{U_4^{(l-1)}} & \ghost{U_3^{(l-1)}} & \qw & \qw & \qw & \ctrl{1} & \qw & \qw & \qw & \ghost{U_3^{(l-1)}} & \ghost{U_4^{(l-1)}} & \qw & \qw & \qw \\
        \lstick{a_4} & \qw & \qw & \qw & \ghost{U_4^{(l-1)}} & \qw & \qw & \qw & \qw & \ctrl{0} & \qw & \qw & \qw & \qw & \ghost{U_4^{(l-1)}} & \qw & \qw & \qw \\
         &  &  &  &  &  &  &  &  &  &  &  &  &  &  &  &  &
    }
}

\newcommand{\qcExampleRecursiveLevelTwoExample}{
    \Qcircuit @C=.5em @R=.7em {
        \lstick{\qubit{x}} & {/} \qw & \sgate{f_1}{2} & \sgate{f_2}{1} & \qw & \sgate{f_2}{1} & \sgate{f_1}{2}\barrier{3} & \sgate{f_3}{1} & \qw & \sgate{f_3}{1}\barrier{2} & \sgate{f_4}{1} & \qw\barrier{1} & \qw & \qw\barrier{1} & \sgate{f_4}{1}\barrier{2} & \sgate{f_3}{1} & \qw & \sgate{f_3}{1}\barrier{3} & \sgate{f_1}{2} & \sgate{f_2}{1} & \qw & \sgate{f_2}{1} & \sgate{f_1}{2} & \qw \\
        \lstick{a_1} & \qw & \qw & \targ & \ctrl{1} & \targ & \qw & \targ & \ctrl{1} & \targ & \targ & \qw & \ctrl{1} & \qw & \targ & \targ & \ctrl{1} & \targ & \qw & \targ & \ctrl{1} & \targ & \qw & \qw \\
        \lstick{a_2} & \qw & \targ & \qw & \ctrl{1} & \qw & \targ & \qw & \targ & \qw & \qw & \qw & \ctrl{1} & \qw & \qw & \qw & \targ & \qw & \targ & \qw & \ctrl{1} & \qw & \targ & \qw \\
        \lstick{a_3} & \qw & \qw & \qw & \targ & \qw & \qw & \qw & \qw & \qw & \qw & \qw & \ctrl{0} & \qw & \qw & \qw & \qw & \qw & \qw & \qw & \targ & \qw & \qw & \qw
    }
}

\newcommand{\qcRearrangeOrigin}{
    \Qcircuit @C=.5em @R=.7em {
        \lstick{\qubit{x}} & {/} \qw & \qw &  &  &  &  & \qw & \qw & \qw & \sgate{f_3}{1} & \sgate{f_2}{2} & \sgate{f_1}{3} & \qw & \sgate{f_4}{2} & \sgate{f_5}{1} & \qw & \qw & \qw &  &  &  \\
        \lstick{a_1} & \qw & \qw &  &  &  &  & \qw & \ctrl{1} & \qw & \targ & \qw & \qw & \qw & \qw & \targ & \qw & \ctrl{1} & \qw &  &  &  \\
        \lstick{a_2} & \qw & \qw &  & \cdots &  &  & \qw & \ctrl{1} & \qw & \qw & \targ & \qw & \qw & \targ & \qw & \qw & \ctrl{1} & \qw &  & \cdots &  \\
        \lstick{a_3} & \qw & \qw &  &  &  &  & \qw & \ctrl{1} & \qw & \qw & \qw & \targ & \qw & \qw & \qw & \qw & \targ & \qw &  &  &  \\
        \lstick{a_4} & \qw & \qw &  &  &  &  & \qw & \targ & \qw & \qw="le" & \qw & \qw & \qw & \qw & \qw & \qw & \qw & \qw\gategroup{1}{11}{5}{16}{1em}{--} &  &  &
    }
}

\newcommand{\qcRearrangeExpanded}{
    \Qcircuit @C=.3em @R=1.1em {
        \lstick{\lvert x_1\rangle} & \qw & \qw &  &        &  &  & \qw & \qw & \qw & \ctrl{1}="qe" & \qw & \qw & \qw & \ctrl{2} & \qw & \ctrl{1} & \ctrl{5} & \qw & \ctrl{3} & \qw & \ctrl{5} & \ctrl{3} & \qw & \qw & \qw & \qw &  &  &  &  \\
        \lstick{\lvert x_2\rangle} & \qw & \qw &  &        &  &  & \qw & \qw & \qw & \ctrl{3} & \ctrl{1} & \ctrl{3} & \qw & \qw & \qw & \ctrl{4} & \qw & \ctrl{2} & \qw & \ctrl{1} & \qw & \qw & \ctrl{2} & \qw & \qw & \qw &  &  &  &  \\
        \lstick{\lvert x_3\rangle} & \qw & \qw &  & \cdots &  &  & \qw & \qw & \qw & \qw & \ctrl{2} & \qw & \qw & \ctrl{3} & \ctrl{1} & \qw & \qw & \qw & \qw & \ctrl{4} & \qw & \qw & \ctrl{2} & \ctrl{1} & \qw & \qw &  &  & \cdots &  \\
        \lstick{\lvert x_4\rangle} & \qw & \qw &  &        &  &  & \qw & \qw & \qw & \qw & \qw & \qw & \ctrl{1} & \qw & \ctrl{2} & \qw & \qw & \ctrl{3} & \ctrl{3} & \qw & \qw & \ctrl{2} & \qw & \ctrl{1} & \qw & \qw & \qw &  &  &  \\
        \lstick{a_1}               & \qw & \qw &  &        &  &  & \qw & \ctrl{1} & \qw & \targ & \targ & \targ & \targ & \qw & \qw & \qw & \qw & \qw & \qw & \qw & \qw & \qw & \targ & \targ & \qw & \ctrl{1} & \qw &  &  &  \\
        \lstick{a_2}               & \qw & \qw &  & \cdots &  &  & \qw & \ctrl{1} & \qw & \qw & \qw & \qw & \qw & \targ & \targ & \targ & \targ & \qw & \qw & \qw & \targ & \targ & \qw & \qw & \qw & \ctrl{1} & \qw &  & \cdots &  \\
        \lstick{a_3}               & \qw & \qw &  &        &  &  & \qw & \ctrl{1} & \qw & \qw & \qw & \qw & \qw & \qw & \qw & \qw & \qw & \targ & \targ & \targ & \qw & \qw & \qw & \qw & \qw & \targ & \qw &  &  &  \\
        \lstick{a_4}               & \qw & \qw &  &        &  &  & \qw & \targ & \qw & \qw & \qw & \qw & \qw & \qw & \qw & \qw & \qw & \qw & \qw & \qw & \qw & \qw & \qw & \qw & \qw & \qw & \qw\gategroup{1}{11}{8}{25}{1em}{--} &  &  &
    }
}

\newcommand{\qcRearrangeAfter}{
    \Qcircuit @C=.3em @R=1.1em {
        \lstick{\lvert x_1\rangle} & \qw & \qw &  &        &  &  & \qw & \qw      & \qw & \ctrl{1} & \qw       & \ctrl{5} & \qw      & \ctrl{2} & \qw      & \ctrl{1} & \qw      & \ctrl{5} & \ctrl{3} & \qw      & \ctrl{3}  & \qw      & \qw & \qw      &  &  &  &  \\
        \lstick{\lvert x_2\rangle} & \qw & \qw &  &        &  &  & \qw & \qw      & \qw & \ctrl{3} & \qw       & \qw      & \ctrl{3} & \qw      & \qw      & \ctrl{4} & \ctrl{2} & \qw      & \qw      & \ctrl{1} & \qw       & \qw      & \qw & \qw      &  &  &  &  \\
        \lstick{\lvert x_3\rangle} & \qw & \qw &  & \cdots &  &  & \qw & \qw      & \qw & \qw      & \ctrl{1}  & \qw      & \qw      & \ctrl{3} & \qw      & \qw      & \qw      & \qw      & \qw      & \ctrl{4} & \qw       & \ctrl{1} & \qw & \qw      &  &  & \cdots &  \\
        \lstick{\lvert x_4\rangle} & \qw & \qw &  &        &  &  & \qw & \qw      & \qw & \qw      & \ctrl{2}  & \qw      & \qw      & \qw      & \ctrl{1} & \qw      & \ctrl{3} & \qw      & \ctrl{3} & \qw      & \ctrl{2}  & \ctrl{1} & \qw & \qw      & \qw &  &  &  \\
        \lstick{a_1}               & \qw & \qw &  &        &  &  & \qw & \ctrl{1} & \qw & \targ    & \qw       & \qw      & \targ    & \qw      & \targ    & \qw      & \qw      & \qw      & \qw      & \qw      & \qw       & \targ    & \qw & \ctrl{1} & \qw &  &  &  \\
        \lstick{a_2}               & \qw & \qw &  & \cdots &  &  & \qw & \ctrl{1} & \qw & \qw      & \targ     & \targ    & \qw      & \targ    & \qw      & \targ    & \qw      & \targ    & \qw      & \qw      & \targ     & \qw      & \qw & \ctrl{1} & \qw &  & \cdots &  \\
        \lstick{a_3}               & \qw & \qw &  &        &  &  & \qw & \ctrl{1} & \qw & \qw      & \qw       & \qw      & \qw      & \qw      & \qw      & \qw      & \targ    & \qw      & \targ    & \targ    & \qw       & \qw      & \qw & \targ    & \qw &  &  &  \\
        \lstick{a_4}               & \qw & \qw &  &        &  &  & \qw & \targ    & \qw & \qw      & \qw       & \qw      & \qw      & \qw      & \qw      & \qw      & \qw      & \qw      & \qw      & \qw      & \qw       & \qw      & \qw & \qw      & \qw\gategroup{1}{11}{8}{23}{1em}{--} &  &  &
    }
}

\newcommand{\qcGroverFramework}{
    \hspace{8em}
    \Qcircuit @C=1em @R=.7em {
        \lstick{n \text{ variable qubits } \qubit{0}} & {/} \qw &
        \gate{\op{H}^{\otimes n}} &
        \multigate{2}{G_1} & \qw &
        \multigate{2}{G_2} & \qw & \qw & \qw & \multigate{2}{G_k} & \qw &
        \meter \\
        & & & & &   & & \cdots & & & &\\
        \lstick{m \text{ ancilla qubits }\qubit{0}} & {/} \qw &
        \qw &
        \ghost{G_1} & \qw & \ghost{G_2} & \qw & \qw & \qw & \ghost{G_k} & \qw & \qw
    }
}

\newcommand{\qcGroverFrameworkRandomized}{
    \hspace{8em}
    \Qcircuit @C=1em @R=.7em {
        \lstick{n \text{ variable qubits } \qubit{0}} & {/} \qw &
        \gate{\op{H}^{\otimes n}} &
        \multigate{2}{G_1} & \qw &
        \multigate{2}{G_2} & \qw & \qw & \qw & \multigate{2}{G_K} & \qw &
        \meter \\
        & & & & &   & & \cdots & & & &\\
        \lstick{m \text{ ancilla qubits } \qubit{0}} & {/} \qw &
        \qw &
        \ghost{G_1} & \qw & \ghost{G_2} & \qw & \qw & \qw & \ghost{G_K} & \qw & \qw
    }
}

\date{}
\begin{document}

\maketitle

\begin{abstract}
Nonlinear boolean equation systems play an important role in a wide range
of applications. Grover's algorithm is one of the best-known quantum
search algorithms in solving the nonlinear boolean equation system on
quantum computers. In this paper, we propose three novel techniques to
improve the efficiency under Grover's algorithm framework. A W-cycle
circuit construction introduces a recursive idea to increase the solvable
number of boolean equations given a fixed number of qubits. Then, a greedy
compression technique is proposed to reduce the oracle circuit depth.
Finally, a randomized Grover's algorithm randomly chooses a subset of
equations to form a random oracle every iteration, which further reduces
the circuit depth and the number of ancilla qubits. Numerical results on
boolean quadratic equations demonstrate the efficiency of the proposed
techniques.
\end{abstract}


\section{Introduction}

A nonlinear boolean equation system with $R$ equations and $n$ boolean
variables admits,
\begin{equation} \label{eq:original-problem}
    f_{j}(x) = \bigoplus_{i_1, i_2, \dotsc, i_n = 0}^1 c_{i_1, i_2,
    \dotsc, i_n} x_1^{i_1} x_2^{i_2} \cdots x_n^{i_n} = 0, \quad
    j = 1, \dotsc, R,
\end{equation}
where $c_{i_1,i_2,\dotsc,i_n} \in \bbF_2$ are boolean coefficients, $(x_1,
\dotsc, x_n) \in \bbF_2^n$ is the vector of boolean variables, the
$\oplus$ sign is the \op{XOR} (logical exclusive disjunction) function
which is the addition operation on $\bbF_2$, and the multiplication is the
\op{AND} (logical conjunction) operation. Nonlinear boolean equation
systems appear in a wide range of applications, including but not limited
to logic synthesis~\cite{micheli1994synthesis}, switching
networks~\cite{mccluskey1956minimization}, cryptography, etc. Among these
applications, it has become an important aspect of cryptography. Most
cryptography algorithms, for example, RSA~\cite{RSA1978} and
ECC~\cite{ECC1987}, are considered secure based on their negligible
success probability for attacks with bounded computational resources,
whose complexity comes from solving nonlinear boolean equations. There has
already been rich research on attacks with low nonlinearity cases, which
is vulnerable to linear approximation attacks~\cite{LinearCryptoDES}.
Further, via low order approximation~\cite{LowOrderApproxCipher,
FastLowOrderApproxCipher}, attacks are advanced to a higher level, where
the difficulty comes from solving highly nonlinear boolean equations.

As many claimed the achievement of quantum
supremacy~\cite{GoogleQuantumSupermacy, USTCQuantumSupermacy}, the quantum
computer becomes an attractive platform to address these exponentially
scaling problems, including solving boolean equations and cryptography
attacks~\cite{Shor, HHLAES2018}. Grover's algorithm~\cite{Grover} provides a
$\bigO{\sqrt N}$ algorithm for an unsorted database searching problem.
Grover's algorithm also makes it possible to address any numerical
optimization problem as long as it can be formulated as a search problem.
It reaches the asymptotic optimality and is the best-known quantum
algorithm for problems when classical algorithms cannot perform better
than brute force searching~\cite{QuantumSolveNP}.
F\"urer~\cite{GroverOptimization} briefed the procedure using Grover's
algorithm as a framework for adaptive global optimization algorithms.
Grassl {et al.}~\cite{GroverAES} extended Grover's algorithm to cracking
one of the most famous block ciphers, AES, even though it had been
suggested to be quantum-safe. It points out that, in principle, a quantum
implementation in quantum mechanics is possible once one embeds operations
in the cipher into permutations, which is reversible and can be viewed as
a subset of all unitary operations.

Finding solutions to nonlinear boolean equations~\eqref{eq:original-problem}
can also be perceived as a searching problem through the truth table of the
boolean function.  A translation from the equations to the quantum circuit
oracle is needed to adapt Grover's algorithm for nonlinear boolean equations.
Generally, the oracle can be perceived as a circuit representing reversible
functions, and much research in recent years has focused on the synthesis
of it.  For example, Shende {et al.}~\cite{SPMH03} link the synthesis of
reversible circuits with the truth table's permutation representation.
Given the truth table of the functions, Shende {et al.}~\cite{SPMH03}
proves that it can be realized with at most one ancilla qubit using only
\op{NOT}, \op{CNOT} and \op{CCNOT} gates (See Section \ref{sec:Preliminary}
for gate definitions) by using the truth table's permutation representation.
The optimal reversible circuits for small samples are examined using
searching algorithms as well.  Then, Brodsky~\cite{Bro04} further derives
the upper bound of the circuit size for certain types of Boolean functions
and explicitly constructs the circuit utilizing the cycle representation of
the permutation.  Besides, Travaglione {et al.}~\cite{BMHA02} considers the
special case where some qubits are read-only and explores the construction of
reversible circuits in this circumstance.  Focusing on minimizing the circuit
depth, Iwama {et al.}~\cite{KYS02} introduces several local transformation
rules for optimizing \op{CNOT}-based circuits.  The rules provide heuristic
methods to reduce the circuit depth when a full search is infeasible.

In this paper, we propose three novel techniques to improve the efficiency
and reduce the cost of solving the boolean quadratic equations under
Grover's algorithm framework. Although all techniques in this paper are
numerically tested in solving boolean quadratic equations (BQE), they can
be adapted to nonlinear boolean equations of higher algebraic orders.

\begin{itemize}

    \item (W-cycle Oracle) We propose a W-cycle structure in the
    construction of an oracle for the nonlinear boolean equations. Such a
    construction could use fewer qubits at the cost of a deeper circuit.
    This method provides a flexible trade-off between the number of
    required qubits and the circuit depth. The maximum number of boolean
    equations for a fixed number of qubits is carefully calculated. The
    asymptotic circuit depth is also estimated in this paper.

    \item (Oracle Compression) A rearranging technique is introduced to
    significantly reduce the circuit depth. It uses a greedy strategy to
    change the order of some interchangeable \op{NOT} and
    controlled--\op{NOT} gates. Many \op{NOT} and controlled--\op{NOT}
    gates are then canceled with each other, and the overall circuit depth
    is reduced without loss of accuracy.

    \item (Randomized Grover's Algorithm) We propose a randomized Grover's
    algorithm to establish a trade-off between the computational cost and
    success rate. The Grover operator varies randomly in each iteration
    by using only part of the boolean equations. The smaller number of
    boolean equations leads to a much more shallow quantum circuit.

\end{itemize}

Finally, we implement all the above three techniques in IBM
Qiskit~\cite{Qiskit} and apply them to address nonlinear boolean
equations. Numerical results show the efficiency of all three techniques.
Given a quantum computer with 25 qubits, we are able to solve 21 boolean
equations with 20 variables.

\begin{figure}[htb]
    \centering
    \includegraphics[trim={0 1.5cm 0 3.3cm},clip,width=0.8\textwidth]{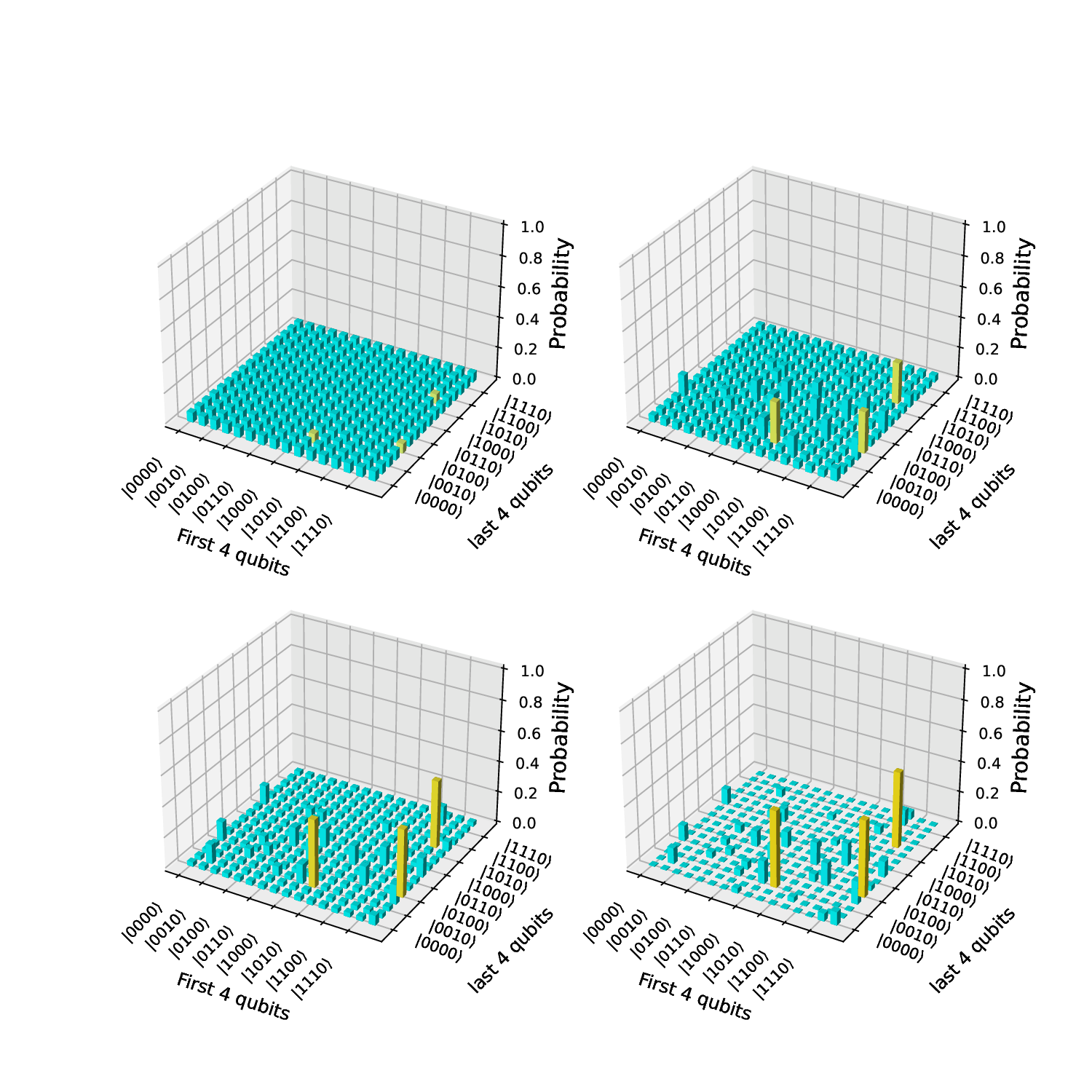}
    \includegraphics[width=0.75\textwidth]{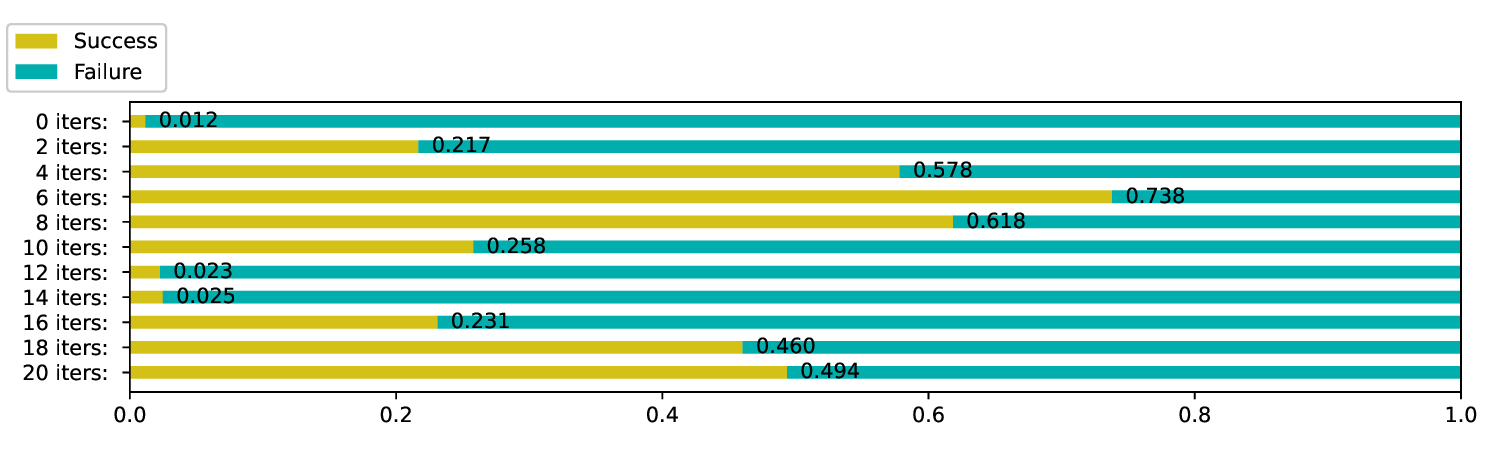}
    \caption{An $8$-qubit example as the illustration of the randomized
    Grover's algorithm. From left to right, top to bottom, the snapshots
    are the probability distribution among all states after $0$, $2$, $4$,
    and $6$ iterations. We observe that the probability distribution is
    concentrated more and more on the correct states during these $6$
    iterations. The $x$-axis
    and $y$-axis represent the first and the last $4$ qubits respectively.
    The $z$-axis is the probability of the corresponding state. Note that
    for visual clarity, the $x$-axis labels and $y$-axis labels are not
    fully listed.
    The bottom bar chart shows the change in the sum of
    probability for correct states at each iteration. The sum reaches
    the maximum at $6$ iterations.
    }

    \label{fig:sample-run-illustration}
\end{figure}

Figure \ref{fig:sample-run-illustration} illustrates how our randomized
Grover's algorithm works via an 8-variable quadratic boolean equations
example. It includes $4$ snapshots of the probability distribution among
all states during the process of iteration. Initiating from a state with
equal probabilities across all potential outcomes, the algorithm is
observed to incrementally amplify the probabilities associated with the
states that correspond to the solutions. Due to the randomized iterations,
probabilities at other states are not reduced consistently but vary at
different states. Grover's iteration differs from a normal iteration scheme
in that there exists some `optimal' number of iterations. The bar chart
provides the success probability at each iteration, where success probability
means the sum of the probabilities over three correct states.

The rest of the paper is organized as follows.
Section~\ref{sec:Preliminary} reviews the vanilla Grover's algorithm.
Section~\ref{sec:Construction} provides details of the technique for
constructing the W-cycle oracle and circuit reorganization.
Section~\ref{sec:Algorithm} explains the idea of randomized Grover's
algorithm and discusses the impacts of various iterative schemes.
Numerical results are demonstrated in Section~\ref{sec:Numerical}.
Finally, Section~\ref{sec:Conclusion} concludes the paper with a
discussion on future works.

\section{Preliminary}
\label{sec:Preliminary}

Grover's algorithm has been proposed and developed for decades. This
section serves as a review of applying Grover's algorithm to solve boolean
equations. We will first introduce notations and quantum circuit diagrams
used throughout the paper in Section~\ref{sec:notation}. Then a brief
review of Grover's algorithms is included in Section~\ref{sec:Grover}.

\subsection{Notations and Diagram}
\label{sec:notation}

Given two 1-qubit orthonormal basis states $\qubit{0}$ and $\qubit{1}$,
any 1-qubit state can be represented as a linear combination, i.e.,
$\qubit{x} = \alpha \qubit{0} + \beta \qubit{1}$, where $\alpha$ and
$\beta$ are complex coefficients satisfying the unit-length constraint,
$\abs{\alpha}^2 + \abs{\beta}^2 = 1$. Similarly, an $n$-qubit state
$\qubit{x}$ can be represented as a linear combination of $n$-qubit basis
states. The $n$-qubit orthonormal basis states are tensor products of $n$
1-qubit basis states and are denoted as $\qubit{0} =
\qubit{0\cdots00}, \qubit{1}=\qubit{0\cdots01}, \dots, \qubit{N-1} =
\qubit{1\cdots11}$ for $N = 2^n$. Alternatively, an $n$-qubit state
$\qubit{x} = \sum_{i=0}^{N-1} \alpha_i \qubit{i}$ is often represented by
its coefficient vector $a = \begin{pmatrix} \alpha_0 & \cdots &
\alpha_{N-1} \end{pmatrix}^\top \in \bbC^N$, where $a$ is of unit length
in 2-norm. The tensor product of two states is denoted as
$\qubit{x}\otimes\qubit{x^\prime} = \qubit{x}\qubit{x^\prime}$ and the
coefficient vector of $\qubit{x} \qubit{x^\prime}$ is the Kronecker
product of their coefficient vectors.

The basic operation in the quantum computer is called the quantum gate,
which manipulates qubits by applying some unitary transformations. Several
basic gates are widely used throughout the paper as our building blocks,
namely, \op{NOT}, \op{CNOT}, \op{MCX}, and \op{MCZ}. The \op{NOT} gate
flips the qubits, i.e., it changes $\qubit{0}$ to $\qubit{1}$ and changes
$\qubit{1}$ to $\qubit{0}$. More explicitly, it exchanges the coefficient
of $\qubit{0}$ and $\qubit{1}$,
\begin{equation*}
    \op{NOT} (\alpha\qubit{0} + \beta\qubit{1})
    = \beta\qubit{0} + \alpha\qubit{1}
    =
    \begin{pmatrix}
        0 & 1 \\
        1 & 0
    \end{pmatrix}
    \begin{pmatrix}
        \alpha\\
        \beta
    \end{pmatrix}.
\end{equation*}
A \op{CNOT} gate is a controlled \op{NOT} gate that acts on two qubits, a
control qubit and a target qubit. The \op{NOT} gate is applied to the
target qubit only if the control qubit is in $\qubit{1}$. For example,
given a 2-qubit state
\begin{equation*}
    \qubit{x} \qubit{x^\prime} = \alpha_{00} \qubit{00}
    + \alpha_{01} \qubit{01} + \alpha_{10} \qubit{10}
    + \alpha_{11} \qubit{11},
\end{equation*}
if we apply a controlled \op{NOT} gate on the first qubit, the state will be
changed to
\begin{equation*}
    \op{CNOT} (\qubit{x}\qubit{x^\prime}) = \alpha_{00} \qubit{00}
    + \alpha_{01} \qubit{01} + \alpha_{11} \qubit{10}
    + \alpha_{10} \qubit{11} =
    \begin{pmatrix}
        1 & 0 & 0 & 0 \\
        0 & 1 & 0 & 0 \\
        0 & 0 & 0 & 1 \\
        0 & 0 & 1 & 0 \\
    \end{pmatrix}
    \begin{pmatrix}
        \alpha_{00}\\
        \alpha_{01}\\
        \alpha_{10}\\
        \alpha_{11}
    \end{pmatrix}.
\end{equation*}
After the \op{CNOT} operation, the 2-qubit system is said to be entangled
since we cannot represent $\op{CNOT}(\qubit{x}\qubit{x^\prime})$ as a
tensor product of two states. The \op{MCX} gate is a multi-controlled
\op{NOT} gate that acts on multiple qubits, specifically multiple control qubits and a
target qubit. It is defined as applying a \op{NOT} gate to the target qubit only if all control qubits are in $\qubit{1}$. When there are two control qubits,
\op{MCX} is also denoted as \op{CCNOT}. The \op{MCZ} gate is a
multi-controlled \op{Z} gate similar to \op{MCX}, except that the 1-qubit
gate been controlled is a \op{Z} gate,
\begin{equation*}
    \op{Z}(\alpha\qubit{0}+\beta\qubit{1})
    = \alpha\qubit{0} - \beta\qubit{1} =
    \begin{pmatrix}
        1 & \\ & -1
    \end{pmatrix}
    \begin{pmatrix}
        \alpha \\ \beta
    \end{pmatrix}.
\end{equation*}
The matrix representation of the \op{MCZ} with one controlling qubit
admits
\begin{equation*}
    \op{MCZ} (\qubit{x}\qubit{x^\prime}) = \alpha_{00} \qubit{00}
    + \alpha_{01} \qubit{01} + \alpha_{10} \qubit{10}
    - \alpha_{11} \qubit{11} =
    \begin{pmatrix}
        1 & 0 & 0 & 0 \\
        0 & 1 & 0 & 0 \\
        0 & 0 & 1 & 0 \\
        0 & 0 & 0 & -1 \\
    \end{pmatrix}
    \begin{pmatrix}
        \alpha_{00}\\
        \alpha_{01}\\
        \alpha_{10}\\
        \alpha_{11}
    \end{pmatrix}.
\end{equation*}
For \op{MCZ}, the sign of the coefficient is flipped if all qubits,
including control qubits and the target qubit, are in $\qubit{1}$. In
other words, the roles of control qubits and the target qubit could be
swapped without affecting the outcomes. Hence, in latter quantum circuit
diagrams, the \op{MCZ} is depicted without distinguishing the control
qubits and the target qubit.

\begin{figure}[htb]
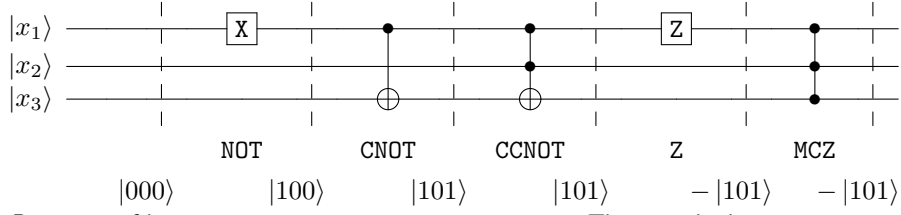

    \centerline{\qcExampleNOTGates}
    \caption{Diagrams of basic quantum gates in quantum circuits. The \op{X}
    in the box represents a \op{NOT} gate. For \op{CNOT} and \op{CCNOT},
    the solid dots are on the control qubits, and $\oplus$ is on the
    target qubit. The \op{Z} in the box represents a \op{Z} gate, whereas the
    \op{MCZ} gate is represented by solid dots on all the qubits it acts
    on. A toy example with an initial state
    $\qubit{000}$ is given below the circuit. }
    \label{fig:gate-example}
\end{figure}

Figure~\ref{fig:gate-example} illustrates \op{NOT}, \op{CNOT}, \op{CCNOT},
\op{Z}, and \op{MCZ} quantum gates in a quantum circuit with 3 qubits.
These circuit diagrams will be repeatedly used in this paper.
In this figure, we give a toy example of a quantum state
starting from $\qubit{000}$. The first \op{NOT} gate flips the first qubit
from $\qubit{0}$ to $\qubit{1}$. Then \op{CNOT} flips the third qubit.
\op{CCNOT} does not change the state since not all control qubits are in
\qubit{1} state. The following \op{Z} gate on the first qubit adds a
negative sign. Finally, the \op{MCZ} gate acts on $\qubit{101}$ and the
state remains the same.

A Hardmard gate \op{H} is a single qubit gate that transforms the basis states into
a superposition of the basis states,
\begin{equation*}
    \op{H} (\alpha\qubit{0} + \beta\qubit{1})
    = \frac\alpha{\sqrt2}\left(\qubit{0} + \qubit{1}\right) + \frac\beta{\sqrt2}\left(\qubit{0} - \qubit{1}\right)
    = \frac1{\sqrt2}
    \begin{pmatrix}
        1 & 1 \\
        1 & -1
    \end{pmatrix}
    \begin{pmatrix}
        \alpha\\
        \beta
    \end{pmatrix}.
\end{equation*}
One of its important properties is that it can be used to create a uniform superposition of all basis states on multiple qubits that are initialized to $\qubit{0}$,
\begin{equation*}
    \op{H}^{\otimes n}\left(\qubit{0}\qubit{0}\cdots\qubit{0}\right) = \frac1{\sqrt{2^n}}\sum_{i=0}^{2^n-1}\qubit{i}.
\end{equation*}
This is a crucial initialization step in many quantum algorithms, including Grover's algorithm.
\subsection{Grover's Algorithm}
\label{sec:Grover}

Grover's algorithm is a quantum unstructured search algorithm. Given an
oracle circuit $O$ producing a different output for a particular input,
Grover's algorithm starts from an equal probability for all possible
inputs and applies an iterative procedure to amplify the probability of
particular desired inputs. Finally, a measurement would result in one of
the particular inputs with high probability. In this section, we review
Grover's algorithm and take a boolean function solver as an example.

For an $n$-variable boolean function as in \eqref{eq:original-problem},
there are $N = 2^n$ different inputs for $x_1, \dots, x_n$, and we adapt
$n$ qubits to represent the probabilities for all inputs. Another $m$
qubits, known as the oracle workspace or ancillae, are reserved for oracle
circuits. A different construction of the oracle circuits for boolean
functions results in different $m$. The vanilla construction, as in
Section~\ref{sec:basic-oracle-construction}, requires $m$ to be at least
the number of equations in \eqref{eq:original-problem}, $R$. At the
beginning of Grover's algorithm, the first $n$ qubits are initialized as
equal probabilities for all inputs, i.e.,
\begin{equation} \label{eq:psi}
    \qubit{\psi} = \frac{1}{\sqrt{N}} \sum_{i=0}^{N-1} \qubit{i}
\end{equation}
via applying $n$ Hardmard gate \op{H} to $\qubit{0} = \qubit{0}^{\otimes
n}$. The latter $m$ qubits are kept in $\qubit{0} = \qubit{0}^{\otimes m}$.

\begin{figure}[htb]
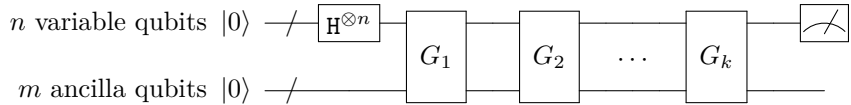

    \centerline{\qcGroverFramework}
    \caption{Grover's algorithm. In vanilla Grover's algorithm, all $G_i$s
    are $G = WO$ for $O$ and $W$ being defined in \eqref{eq:oracle-action}
    and \eqref{eq:householder-action} respectively.}
    \label{fig:Grover-framework}
\end{figure}

After the initialization, Grover's algorithm applies a sequence of Grover
iterations. Each Grover iteration applies two circuits: an oracle $O$ and
a Householder-like diffusion operator $W$. The oracle $O$ is application-dependent and often
treated as a black box in Grover's algorithm. In general unstructured
search problems, the oracle flips the sign of desired solutions while
keeping others unchanged. In solving boolean functions, the oracle flips the
sign of $\qubit{x}$ if $x$ is a solution of \eqref{eq:original-problem}.
More precisely, the action of the oracle $O$ is
\begin{equation} \label{eq:oracle-action}
    O: \qubit{x}\qubit{0}^{\otimes m} \mapsto
    (-1)^{g(x)} \qubit{x}\qubit{0}^{\otimes m},
    \qquad x=0, 1, \dotsc, N-1,
\end{equation}
where $g(x)$ is the solution indicator function as
\begin{equation*}
    g(x) =
    \begin{cases}
        1 & \text{if $x$ is a solution of \eqref{eq:original-problem},} \\
        0 & \text{otherwise.}
    \end{cases}
\end{equation*}
The quantum circuits for the oracle $O$ of the boolean functions are
detailed in Section~\ref{sec:Construction}, and efficient constructions
are also proposed therein. Importantly, we emphasize that the ancilla
qubits have to remain in $\qubit{0} = \qubit{0}^{\otimes m}$ after the
oracle circuit to avoid any side effects. The Householder-like diffusion operation $W
= 2\qubit{\psi}\bra{\psi} - I$ is then applied to the first $n$ qubits,
whose action on a general state $\sum_i c_i\qubit{i}$ admits
\begin{equation} \label{eq:householder-action}
    W: \sum_i c_i \qubit{i} \mapsto
    \sum_i \left( - c_i + 2 \langle c\rangle \right) \qubit{i},
\end{equation}
where $\qubit{\psi}$ is as defined in \eqref{eq:psi}, and $\langle
c\rangle$ is the average of $\{ c_i\}$. This is a crucial step to
amplifies the probabilities of solutions.
Figure~\ref{fig:Grover-framework} shows an illustration of Grover's
algorithm.

Though Grover's algorithm carries an iterative procedure, it is different
from classical iterative methods. Instead of converging to a fixed point
of the iterative mapping, Grover's algorithm conducts a fixed number of
iterations. We take the boolean functions \eqref{eq:original-problem} as
an example. All $N$ possible boolean variables are split into solutions
and non-solutions of \eqref{eq:original-problem} and the number of
solutions is denoted as $M$. The set of all solutions is denoted as
$\calS$. The initial state $\qubit{\psi}$ as in \eqref{eq:psi} could be
rewritten as
\begin{equation*}
    \qubit{\psi} = \cos \frac{\theta}{2} \qubit{\alpha}
    + \sin \frac{\theta}{2} \qubit{\beta},
\end{equation*}
where $\qubit{\alpha}$ and $\qubit{\beta}$ are the superpositions of
non-solutions and solutions, i.e.,
\begin{equation*}
    \qubit{\alpha} = \frac{1}{\sqrt{N-M}} \sum_{x \notin \calS}\qubit{x}
    \quad , \quad
    \qubit{\beta} = \frac{1}{\sqrt{M}} \sum_{x \in \calS}\qubit{x},
\end{equation*}
and $\theta$ is determined by $\cos \frac\theta2 = \sqrt{\frac{N-M}N}$. The
action of $G$ on $\qubit{\psi}$ obeys \footnote{Actually, $G = WO$ is
applied to all $n+m$ qubits, i.e., $O$ is applied to $\qubit{\psi}
\qubit{0}^{\otimes m}$, and $W$ is applied to the first $n$ qubits. For
the sake of notations, we omit the $m$ ancilla qubits in
\eqref{eq:apply-first-G} and \eqref{eq:apply-k-times-of-G}.}
\begin{equation} \label{eq:apply-first-G}
    G \qubit{\psi} = W O \qubit{\psi}
    = W \left(\cos{\frac\theta2}\qubit{\alpha}
    - \sin{\frac\theta2}\qubit{\beta}\right)
    = \cos{\frac{3\theta}2}\qubit{\alpha}
    + \sin{\frac{3\theta}2}\qubit{\beta}.
\end{equation}

\begin{figure}[htb]
    \centering
    \begin{tikzpicture}[x=0.9cm,y=0.9cm]
        \usetikzlibrary {arrows.meta}
        \usetikzlibrary{angles,quotes}
        \coordinate (O) at (0,0);
        \coordinate (A) at (0,4);
        \coordinate (B) at (2.828, 2.828);
        \coordinate (C) at (3.86, 1.035);
        \coordinate (D) at (4,0);
        \coordinate (E) at (3.86, -1.035);

        \draw[-{Stealth[length=2.5mm]}] (O) -- (A) node[pos=1.0, right]{$\qubit{\beta}$};
        \draw[-{Stealth[length=2.5mm]}] (O) -- (D) node[pos=1.0, right]{$\qubit{\alpha}$};
        \draw[-{Stealth[length=2.5mm]},line width=1pt] (O) -- (B) node[pos=1.0, above]{$G\qubit{\psi}$};
        \draw[-{Stealth[length=2.5mm]},line width=1pt] (O) -- (C) node[pos=1.0, above right]{$\qubit{\psi}$};
        \draw[-{Stealth[length=2.5mm]},line width=1pt] (O) -- (E) node[pos=1.0, right]{$O\qubit{\psi}$};
        \draw[-, dashed] (B) -- (E) node[pos=1.0, right]{};
        \draw[-, dashed] (C) -- (E) node[pos=1.0, right]{};

        \pic [draw, -, "$\theta$", angle eccentricity=1.5, angle radius=35] {angle = C--O--B};
        \pic [draw, -, "$\theta / 2$", angle eccentricity=1.5, angle radius=35] {angle = D--O--C};
        \pic [draw, -, "$\theta / 2$", angle eccentricity=1.5, angle radius=35] {angle = E--O--D};
    \end{tikzpicture}
    \caption{The effect of an iteration of Grover's algorithm. From the
    initial state $\qubit{\psi}$, the oracle $O$ flips the sign of the
    solution to $g(x)=1$, as shown in \eqref{eq:oracle-action}.  Then, the
    diffusion operation $W$ makes a reflection of the state
    $O\qubit{\psi}$ by $\qubit{\psi}$.  This single step of iteration rotates
    $\qubit{\psi}$ to $G\qubit{\psi}$ by an angle of $\theta$. This step
    reduces the probability of any state in $\qubit{\alpha}$ (i.e., wrong
    solution) being the outcome of a measurement, which is $|\braket{\psi
    | \alpha}|^2$.  Note that rotating over $\qubit{\beta}$, i.e., repeating
    the iteration more than $K$ times computed in
    \eqref{eq:GroverIterNum}, will contradictorily increase the
    probability of any state in $\alpha$ being the outcome of a
    measurement.}
    \label{fig:Grover-Rotation}
\end{figure}
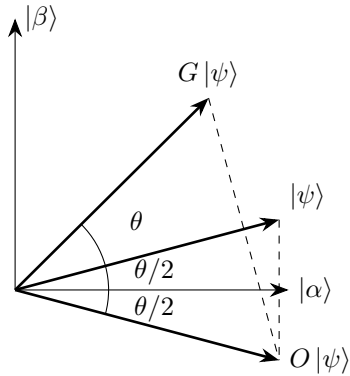

As illustrated in Figure~\ref{fig:Grover-Rotation}, the action of
$G$ could be understood in a geometrical way. Applying the oracle $O$
flips the sign of $\qubit{\alpha}$. Then, the Householder-like operation
$W$ makes the state $O \qubit{\psi}$ mirrored with respect to
$\qubit{\psi}$, and results in a state $G\qubit{\psi}$ have greater
overlapping with $\qubit{\beta}$. Applying Grover iteration $k$ times
leads to
\begin{equation} \label{eq:apply-k-times-of-G}
    G^{k}\qubit{\psi} = (WO)^{k}\qubit{\psi} =
    \cos\left(\frac{2k+1}{2} \theta \right) \qubit{\alpha} +
    \sin\left(\frac{2k+1}{2} \theta \right) \qubit{\beta}.
\end{equation}
Repeating the iteration for
\begin{equation} \label{eq:GroverIterNum}
    K = \mathrm{round} \left( \frac{\arccos\sqrt{M/N}}{\theta} \right)
\end{equation}
times makes $G^K \qubit{\psi}$ sufficiently close to $\qubit{\beta}$ and
the measurement would fall into one of the solutions with a probability
$\sin^2\left( \frac{2K+1}2\theta \right)$, which is close to $1$.

\section{Efficient Oracle Construction}
\label{sec:Construction}

The oracle for boolean functions can be constructed either in a
qubit-efficient way or in a depth-efficient way. We first show both the
qubit-efficient and depth-efficient vanilla constructions of the oracle in
Section~\ref{sec:basic-oracle-construction}. In
Section~\ref{sec:recursive-oracle-construction}, we propose a novel
recursive oracle construction, which exploits the maximum possible number
of equations for a given number of ancillae. Finally, a rearranging
technique is introduced in Section~\ref{sec:oracle-compression} to further
reduce the depth of the circuit by greedily exploring possible
parallelization.

\subsection{Basic Oracle Construction}
\label{sec:basic-oracle-construction}

We first introduce a straightforward oracle construction for a single
boolean equation. Given a boolean equation with $n$ variables, $f(x) = 0$,
we adopt $n+1$ qubits for the oracle, where the first $n$ qubits represent
boolean variables and the last ancilla qubit is used to track the value of
$f(x)$. Without loss of generality, we assume $f(x)$ is in a sum of
product form, as in \eqref{eq:original-problem}. Then each product term is
implemented by \op{CNOT}, \op{CCNOT}, or \op{MCX}. For example, a
\op{CNOT} controlling from the $i$-th qubit to the ancilla represents the
$x_i$ term; a \op{CCNOT} controlling from the $i$- and $j$-th qubits to
the ancilla represents the $x_i x_j$ term; other product terms with more
than two variables could be implemented by an \op{MCX} controlling from
qubits corresponding to the variables therein to the ancilla. After
applying these gates for all product terms in $f(x)$, the ancilla is in
$\qubit{0}$ if $\qubit{x}$ satisfies $f(x) = 0$, and in $\qubit{1}$ if
$\qubit{x}$ satisfies $f(x) = 1$. Then we apply an \op{X} gate followed by
a \op{Z} gate on the ancilla such that a negative sign is applied for all
$\qubit{x}$ solves $f(x) = 0$. Finally, to reset the ancilla qubit to
$\qubit{0}$ after the oracle as in \eqref{eq:oracle-action}, all previous
gates except the last \op{Z} gate are applied one more time. A detailed
quantum circuit of the oracle for $f(x) = x_1 \oplus x_1 x_2 = 0$ as well
as its abbreviation are given in Figure~\ref{fig:sample-equation-circuit}.
In the rest of the paper, the abbreviated quantum circuit is called the
\emph{function-controlled-not} gate or \emph{$f_i$-controlled-not} gate
for a particular boolean function $f_i$. In all later quantum circuit
diagrams, we will only draw the abbreviated diagram instead of complex
circuits.

\begin{figure}[htb]
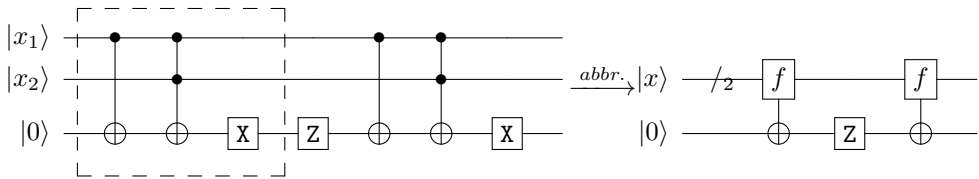

    \centerline{\qcExampleFunctionControl}
    \caption{Oracle circuit for boolean equation: $f(x) = x_1 \oplus x_1
    x_2 =0$.  The extra \op{X} gate ensures the ancilla in state $\qubit{1}$
    if $f(x)=0$.}
    \label{fig:sample-equation-circuit}
\end{figure}

When the boolean equation system involves more than one equation, i.e.,
$R$ equations in \eqref{eq:original-problem} are denoted as $f_1, \dots,
f_R$, there are two basic constructions for the oracle: 1) multiplying
equations together and treating them as a single boolean equation; 2)
adopting $R$ ancilla qubits to track the values of all boolean equations.
We briefly introduce both basic constructions and discuss their drawback.
Then a novel recursive construction for oracle is proposed in
Section~\ref{sec:recursive-oracle-construction}.

For a set of equations $f_i(x) = 0, i = 1, \dotsc, R$, one can multiply
them together as
\begin{equation} \label{eq:multiple-all-equations-into-one}
    f(x) \triangleq (f_1(x) \oplus 1) (f_2(x) \oplus 1) \cdots
    (f_R(x) \oplus 1) \oplus 1.
\end{equation}
A solution of $f(x^*) = 0$ for $f(x)$ as in
\eqref{eq:multiple-all-equations-into-one} requires that $f_i(x^*) = 0$
for all $i = 1, \dots, R$. Hence solving $f(x) = 0$ is equivalent to solve
the boolean equation system \eqref{eq:original-problem}. Then $f(x)$ is
transformed to a sum of product form and the oracle is constructed as we
discussed above. For example, given a quadratic boolean equation system
\begin{equation*}
    \begin{split}
        f_1(x) = & x_1 \oplus x_1 x_2 = 0, \\
        f_2(x) = & x_3 x_4 = 0, \\
        f_3(x) = & x_1 x_4 =0, \\
        f_4(x) = & x_2 \oplus x_3 \oplus x_4 = 0, \\
    \end{split}
\end{equation*}
the sum of product form of $f(x)$ admits
\begin{equation*}
    f(x) = x_1 x_2 x_3 \oplus x_1 x_3 x_4 \oplus x_2 x_3 x_4
    \oplus x_1 x_2 \oplus x_1 x_3 \oplus x_1 x_4
    \oplus x_3 x_4 \oplus x_1 \oplus x_2 \oplus x_3 \oplus x_4.
\end{equation*}
In transformation to the sum of product form, there are cancellations to
reduce the number of product terms. However, the number of product terms
in $f(x)$ is often found to be much larger than that in $f_i(x)$s. Another
drawback is caused by the \op{MCX}. The number of variables in the product
term in $f(x)$ is larger than that in $f_i(x)$s. Hence the corresponding
\op{MCX} has more controlling qubits. It is generally considered that a
quantum gate operation involving many qubits is more difficult to
implement, and an approximated model~\cite{QuantumTextbook} shows an
exponential growth of the circuit depth to the number of controlling
qubits. The circuit depth of such an oracle construction is much larger
than all later oracle constructions and is often found to be impractical
on current quantum devices.

\begin{figure}[htb]
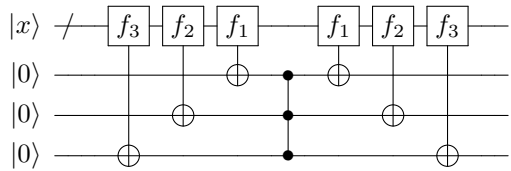

    \centerline{\qcExampleStackAll}
    \caption{A vanilla stack style oracle construction for a boolean
    equation system with three equations.}
    \label{fig:simple-stack-all-oracle-example}
\end{figure}

Another basic construction of the oracle for the boolean equation system
with $R$ equation is to adopt $R$ ancilla qubits to track the outcomes of
all equations. For each $i = 1, \dots, R$, we construct the circuit for
$f_i(x)$ with the $i$-th ancilla qubit. Then $R$ ancilla qubits now keep
the outcome of all boolean equations. The $i$-th ancilla qubit is in
$\qubit{1}$ if $f_i(x) = 0$, and in $\qubit{0}$ otherwise. A solution of
\eqref{eq:original-problem} satisfies all boolean equations, and hence all
ancilla qubits are in $\qubit{1}$. An \op{MCZ} gate is then applied and
flips the sign if all ancilla qubits are in $\qubit{1}$. After applying
the \op{MCZ} gate, we revert all ancilla qubits to their initialized state
by applying all previous single equation circuits in the reversed order.
Figure~\ref{fig:simple-stack-all-oracle-example} provides an example of a
boolean equation system with three equations. In the rest of the paper, we call
such an oracle construction the vanilla stack style. If we assume that
each boolean equation is quadratic and has $P$ product terms, the vanilla
stack style oracle construction has circuit depth $RP + C$ and uses $R$
ancilla qubits, where $C$ is the constant circuit depth of \op{MCZ}
gate.~\footnote{The circuit depth counts the longest sequence of simple
gates, including 1-qubit gates and a few multi-qubit gates, i.e.,
\op{CNOT}, \op{CCNOT}. The \op{MCZ} gate should be decomposed into \op{H}
and \op{MCX} gates, whose depth is denoted as a constant $C$.}

\subsection{Recursive Oracle Construction}
\label{sec:recursive-oracle-construction}

We provide a flexible way of constructing the oracle making a trade-off
between the number of qubits and circuit depth. In our novel construction,
none of the boolean equations are multiplied together. The idea, in a
nutshell, is to build the circuit recursively based on the vanilla stack
style. In the following, we propose the recursive oracle construction in
detail and explore the maximum number of possible boolean equations that
could be constructed on $m$ ancilla qubits. A building block in the
recursive oracle construction is denoted as $\U{\ell}{m}$, where $\ell$
denotes the recursive level and $m$ denotes the number of ancilla qubits.
We will first introduce the quantum circuit $\U{1}{m}$, which is similar
to the vanilla stack style. Then a recursive construction from
$\U{\ell-1}{1}, \U{\ell-1}{2}, \dots, \U{\ell-1}{m-1}$ to $\U{\ell}{m}$ is
proposed. Finally, we construct the oracle quantum circuit based on many
$\U{\ell}{m}$s.

\textbf{Quantum circuit for $\U{1}{m}$.} The idea behind $\U{1}{m}$ is to
fully exploit all ancilla qubits while storing the desired information in
only the last ancilla qubit. Given $m-1$ boolean functions, $f_1(x),
\dots, f_{m-1}(x)$, we first construct a sequence of
function-controlled-not gates controlling from $\qubit{x}$ to the 1st,
2nd, $\dots$, $(m-1)$-th ancilla qubits for boolean function $f_1, f_2,
\dots, f_{m-1}$, respectively. Then an \op{MCX} gate is applied
controlling from ancilla qubits $1, 2, \dots, m-1$ to the $m$-th ancilla
qubits. After these gate operations, the $m$-th ancilla qubit is in
$\qubit{1}$ if $f_1 = 0, f_2 = 0, \dots, f_{m-1} = 0$ are satisfied. Hence
the solutions of $f_i(x) = 0$ for $i=1, 2, \dots, m-1$ are tracked in the
$m$-th ancilla qubit. In order to reuse other ancilla qubits, we apply the
same sequence of function-controlled-not gates in the reversed ordering
and reset ancilla qubits $1, 2, \dots, m-1$ to $\qubit{0}$.
Figure~\ref{fig:level-1-V-shape} depicts an example for $\U{1}{4}$. When
$m = 1$, we use function-controlled-not gate directly on the ancilla.
Hence $\U{1}{1}$ is the same as a single function-controlled-not gate.

\begin{figure}[htb]
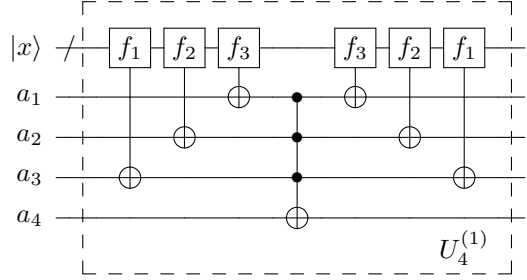

    \centerline{\qcExampleRecursiveLevelOne}
    \caption{Quantum circuit for $\U{1}{4}$. The function-controlled-gate
    is as defined in Figure~\ref{fig:sample-equation-circuit}.}
    \label{fig:level-1-V-shape}
\end{figure}

\textbf{Recursion from $\U{\ell-1}{j}$s to $\U{\ell}{m}$.} Assume we have
constructed all quantum circuits $\U{\ell-1}{1}, \dots, \U{\ell-1}{m-1}$
at previous level $\ell - 1$. We further assume that after applying
$\U{\ell-1}{j}$ only the $j$-th ancilla qubit is affected, and all other
ancilla qubits remain unchanged. The construction of $\U{\ell}{m}$ is as
follows. We first apply $\U{\ell-1}{m-1}, \U{\ell-1}{m-2}, \dots,
\U{\ell-1}{1}$ in order. Then an \op{MCX} gate is applied controlling from
ancilla qubits $1, 2, \dots, m-1$ to the $m$-th ancilla qubit. After
these gate operations, the $m$-th ancilla qubit is in $\qubit{1}$ if all
boolean functions behind $\U{\ell-1}{1}, \dots, \U{\ell-1}{m-1}$ are
satisfied. All these boolean functions are called the boolean functions
behind $\U{\ell}{m}$. Similar as that in $\U{1}{m}$, we apply
$\U{\ell-1}{1}, \U{\ell-1}{2}, \dots, \U{\ell-1}{m-1}$ in order to reset
ancilla qubits $1, 2, \dots, m-1$ to $\qubit{0}$. Importantly, the
ordering of $\U{\ell-1}{1}, \dots, \U{\ell-1}{m-1}$ cannot be changed
since $\U{\ell-1}{j}$s are not commutable. Noticeably, $\U{2}{1}$ is not
defined under our recursion and so are $\U{\ell}{1}$ for $\ell > 1$. In
our recursion, the actual circuits for $\U{\ell}{1}$ are $\U{1}{1}$. For
the sake of notation, $\U{\ell}{m}$ for $\ell > m$ are defined as
$\U{\ell}{m} \triangleq \U{m}{m}$. Figure~\ref{fig:level-k-V-shape}
depicts the construction of $\U{\ell}{4}$ from $\U{\ell-1}{1},
\U{\ell-1}{2}, \U{\ell-1}{3}$.

\begin{figure}[htb]
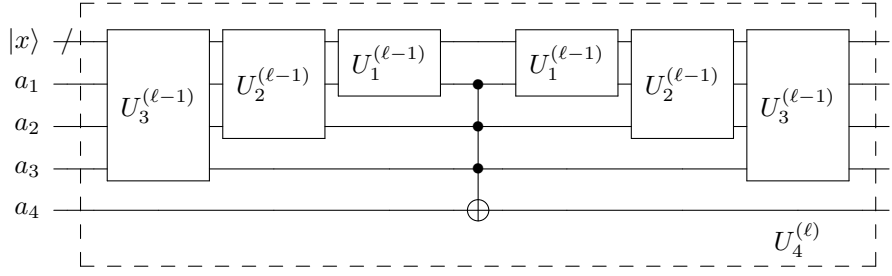

    \centerline{\qcExampleRecursiveLevelK}
    \caption{Quantum circuit for $\U{\ell}{4}$ constructed from
    $\U{\ell-1}{1}$, $\U{\ell-1}{2}$, and $\U{\ell-1}{3}$.}
    \label{fig:level-k-V-shape}
\end{figure}

\textbf{Recursive oracle construction.} Given a recursive level $\ell > 1$
and $m$ ancilla qubits, the quantum circuit for the oracle is composed of
$\U{\ell-1}{1}, \dots, \U{\ell-1}{m}$.~\footnote{For an oracle circuit
with $\ell = 1$, the construction falls back to the vanilla stack style.
Hence we only introduce oracle constructions with $\ell > 1$ in the rest of the
paper.} We first apply $U_m^{(\ell-1)}$ to $U_{1}^{(\ell-1)}$ in order.
Then an \op{MCZ} gate is applied to all ancilla qubits, which introduces
the sign flip for solutions. Finally, $U_1^{(\ell-1)}$ to $U_m^{(\ell-1)}$
are applied in order to reset all ancilla qubits to $\qubit{0}$. The
overall structure for the oracle circuit is similar to that in vanilla
stack style except for function-controlled-gates being replaced by our
recursive circuits $\U{\ell}{j}$s. Figure~\ref{fig:V-shape-circuit}
depicts the oracle circuits with $\ell$ and $m = 4$. Another example is
given in Figure~\ref{fig:level-2-oracle-example} for $\ell = 2$ and $m =
3$, where all $\U{\ell}{j}$s are expanded into function-controlled-gates.
The level $\ell$ oracle construction pseudocode is detailed in
Algorithm~\ref{alg:oraclecircuit} and the recursion pseudocode for
$\U{\ell}{m}$ is in Algorithm~\ref{alg:Ucircuit} for a boolean equation
system with $R$ equations and $m$ ancilla qubits. In these pseudocodes,
the set of boolean equations $\calF = \{f_i\}_{i=1}^R$ is a global
variable and $\calF$.pop() pops out a boolean function from the set. If
the set $\calF$ is empty, $\calF$.pop() then returns a trivial boolean
equation $f(x) \equiv 0$. The $f$-controlled-gate, in this case, is empty,
and no quantum gate is appended to the circuit.

\begin{figure}[htb]
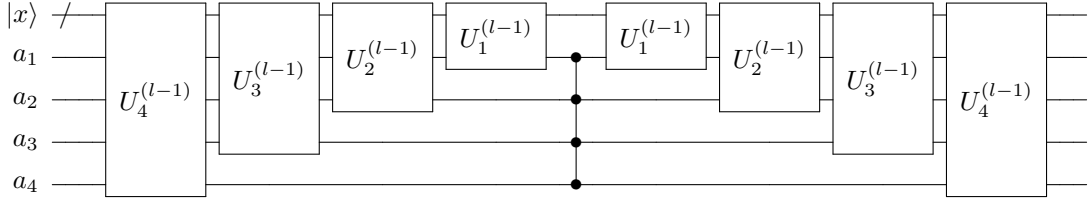

    \centerline{\qcExampleRecursiveCircuit}
    \caption{A level $\ell$ recursive construction of oracle with $m = 4$
    ancilla qubits.} \label{fig:V-shape-circuit}
\end{figure}

\begin{figure}[htb]
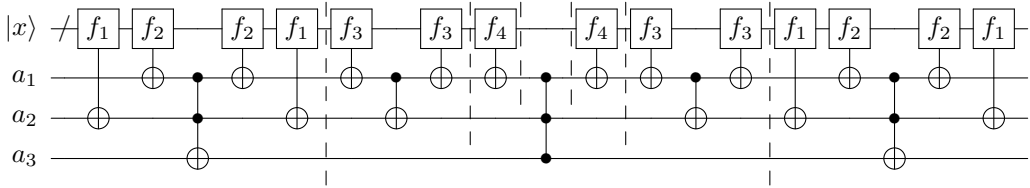

    \centerline{\qcExampleRecursiveLevelTwoExample}
    \caption{A level $\ell = 2$ recursive construction of oracle with $m =
    3$ ancilla qubits for a boolean equation system with 4 equations.}
    \label{fig:level-2-oracle-example}
\end{figure}

\begin{algorithm}
    \caption{$\U{\ell}{m}$ construction.}
    \label{alg:Ucircuit}
    \begin{algorithmic}[1]
        \Require target ancilla qubits $m$, recursive level $\ell$, and
        set of boolean equations $\calF = \{f_i\}_{i=1}^R$.
        \Ensure quantum circuit for $\U{\ell}{m}$.
        \Procedure{Ucircuit}{$\ell$, $m$}
            \IIf {$\ell > m$}
                \Return {\Call{Ucircuit}{$m$, $m$}}
            \EndIIf
            \IIf {$R = 0$}
                \Return {None}
            \EndIIf
            \If {$\ell = 0$}
                \State {$f = \calF$.pop()}
                \State {\Return {$f$-controlled gate controlling $m$-th
                ancilla qubit}}
            \EndIf
            \If {$\ell = 1$}
                \If {$m = 1$}
                    \State {$f = \calF$.pop()}
                    \State {\Return {$f$-controlled gate}}
                \EndIf
                \For {$j = m-1, m-2, \dots, 1$}
                    \State {$f_j = \calF$.pop()}
                    \State {Append $f_j$-controlled gate controlling
                    $j$-th ancilla qubit to output circuit}
                \EndFor
                \State {Append \op{MCX} gate controlling from $1,\dots,m-1$
                ancilla qubits to $m$ ancilla qubit to output circuit}
                \For {$j = 1, 2, \dots, m-1$}
                    \State {Append $f_j$-controlled gate controlling
                    $j$-th ancilla qubit to output circuit}
                \EndFor
            \Else
                \For {$j = m-1, m-2, \dots, 1$}
                    \State {$\U{\ell - 1}{j} = $ \Call{Ucircuit}{$\ell-1$, j}}
                    \State {Append $\U{\ell-1}{j}$ to output circuit}
                \EndFor
                \State {Append \op{MCX} gate controlling from $1,\dots,m-1$
                ancilla qubits to $m$ ancilla qubit to output circuit}
                \For {$j = 1, 2, \dots, m-1$}
                    \State {Append $\U{\ell-1}{j}$ to output circuit}
                \EndFor
            \EndIf
            \State {\Return {output circuit}}
        \EndProcedure
    \end{algorithmic}
\end{algorithm}

\begin{algorithm}
    \caption{Oracle construction.}
    \label{alg:oraclecircuit}
    \begin{algorithmic}[1]
        \Require $m$ ancilla qubits, recursive level $\ell$, and
        set of boolean equations $\calF = \{f_i\}_{i=1}^R$.
        \Ensure oracle quantum circuit.
        \Procedure{Oracle}{$\ell$, $m$}
            \For {$j = m, m-1, m-2, \dots, 1$}
                \State {$\U{\ell - 1}{j} = $ \Call{Ucircuit}{$\ell-1$, j}}
                \State {Append $\U{\ell-1}{j}$ to output circuit}
            \EndFor
            \State {Append \op{MCZ} gate on all ancilla qubits to output
            circuit}
            \For {$j = 1, 2, \dots, m-1, m$}
                \State {Append $\U{\ell-1}{j}$ to output circuit}
            \EndFor
        \EndProcedure
    \end{algorithmic}
\end{algorithm}

Next, we calculate the capacity of the oracle or $\U{\ell}{m}$, i.e., the
maximum number of boolean equations that could be constructed into the
oracle or $\U{\ell}{m}$. Theorem~\ref{thm:Ucapacity} states the capacity
of $\U{\ell}{m}$. Corollary~\ref{cor:oraclecapacity} then sums the
capacities of $\U{\ell - 1}{j}$ for $j = 1, 2, \dots, m-1$ and obtains the
capacity of a level $\ell$ oracle with $m$ ancilla qubits.

\begin{theorem}
    \label{thm:Ucapacity}
    Let $\N{\ell}{m}$ denote the capacity for the level $\ell$ recursive
    quantum circuit $\U{\ell}{m}$ on $m$ ancilla qubits.
    For various scenarios of $\ell$ and $m$, we have
    \begin{equation*}
        \N{\ell}{m} =
        \begin{cases}
            1         & \ell \geq 1, m = 1, \\
            2^{m-2}   & \ell \geq m - 1, m \geq 2, \\
            \sum_{j=0}^\ell \binom{m-2}{j} & \mathrm{Otherwise.}
        \end{cases}
    \end{equation*}
\end{theorem}

\begin{proof}

For various scenarios of $\ell$ and $m$, we have
\begin{equation}
    \label{eq:recursive-formula}
    \N{\ell}{m} =
    \begin{cases}
        1         & \ell = 1, m = 1, \\
        m-1       & \ell = 1, m \geq 2, \\
        \N{m}{m}  & \ell > m, \\
        \sum_{j=1}^{m-1} \N{\ell-1}{j}  & \mathrm{Otherwise.}
    \end{cases}
\end{equation}
The recursive formula~\eqref{eq:recursive-formula} can be obtained
directly from the process of constructing the circuit $\U{\ell}{m}$.
Obviously, we have $\N{\ell}{1} = 1$ and $\N{\ell}{2} = 1$ for all $\ell
\geq 1$.

We first derive the expression of $\N{\ell}{m}$ for $\ell \geq m-2$ and $m
\geq 3$. The claim is that $\N{\ell}{m} = 2^{m-2}$ for $\ell \geq m-2$ and
$m \geq 3$. From the recursive formula~\eqref{eq:recursive-formula}, we
have $\N{1}{3} = 2$, $\N{2}{3} = \N{1}{2} + \N{1}{1} = 2$, $\N{3}{3} =
\N{2}{2} + \N{2}{1} = 2$, and $\N{\ell}{3} = \N{3}{3} = 2$ for all $\ell >
3$. If $\N{j}{n} = 2^{n-2}$ holds for all $j \geq n - 2$ and $n = 3,
\dots, m-1$, then we have
\begin{equation*}
    \begin{split}
        \N{m-2}{m} = & \sum_{j=1}^{m-1} \N{m-3}{j} = \N{m-3}{1}
        + \N{m-3}{2} + \sum_{j = 3}^{m-1} \N{m-3}{j} = 1 + 1
        + \sum_{j=3}^{m-1} 2^{j-2} = 2^{m-2}, \\
        \N{m-1}{m} = & \sum_{j=1}^{m-1} \N{m-2}{j} = \N{m-2}{1}
        + \N{m-2}{2} + \sum_{j = 3}^{m-1} \N{m-2}{j} = 1 + 1
        + \sum_{j=3}^{m-1} 2^{j-2} = 2^{m-2}, \\
        \N{\ell}{m} = \N{m}{m} = & \sum_{j=1}^{m-1} \N{m-1}{j} = \N{m-1}{1}
        + \N{m-1}{2} + \sum_{j = 3}^{m-1} \N{m-1}{j} = 1 + 1
        + \sum_{j=3}^{m-1} 2^{j-2} = 2^{m-2}, \\
    \end{split}
\end{equation*}
for all $\ell > m$. By induction, we prove the claim.

Finally, we derive the expression of $\N{\ell}{m}$ for $2 \leq \ell \leq m-3$
and $m \geq 5$. The recursive formula~\eqref{eq:recursive-formula} could
be written as
\begin{equation} \label{eq:recursive-formula2}
    \N{\ell}{m} = \sum_{j=1}^{m-1} \N{\ell-1}{j}
    = \N{\ell-1}{m-1} + \sum_{j=1}^{m-2} \N{\ell-1}{j}
    = \N{\ell-1}{m-1} + \N{\ell}{m-1}.
\end{equation}
Let us introduce an auxiliary variable and its recursive formula,
\begin{equation} \label{eq:auxiliary-recursive-formula}
    \T{\ell}{m} \triangleq \N{\ell}{m} - \N{\ell-1}{m}
    = \N{\ell-1}{m-1} + \N{\ell}{m-1} - \N{\ell-2}{m-1}
    - \N{\ell-1}{m-1}
    = \T{\ell}{m-1} + \T{\ell-1}{m-1}
\end{equation}
for $2 \leq l \leq m-2$ and $m \geq 4$.
For $m \geq 5$, we have
\begin{equation*}
    \T{m-2}{m} = \N{m-2}{m} - \N{m-3}{m}
    = 2^{m-2} - (\N{m-4}{m-1} + \N{m-3}{m-1})
    = \N{m-3}{m-1} - \N{m-4}{m-1}
    = \T{m-3}{m-1},
\end{equation*}
where the first and last equalities adopt the definition of $\T{\ell}{m}$,
and the second and third equalities adopt the recursive
formula~\eqref{eq:recursive-formula2} and $\N{\ell}{m} = 2^{m-2}$ for
$\ell \geq m-2$. Through a direct calculation, we obtain, $\T{m-2}{m} =
\T{3}{5} = \T{2}{4} = 1$ for all $m \geq 4$. This is one boundary
condition for the recursive formula of $\T{\ell}{m}$. For the other
boundary, we have
\begin{equation*}
    \T{2}{m} = \N{2}{m} - \N{1}{m}
    = \sum_{j=1}^{m-1} \N{1}{j} - (m-1)
    = \sum_{j=2}^{m-2} \N{1}{j}
    = \sum_{j=1}^{m-3} j = \binom{m-2}{2}
\end{equation*}
for $m \geq 4$. Based on the boundary conditions and recursive
formula~\eqref{eq:auxiliary-recursive-formula}, we notice that
$\T{\ell}{m}$ is a part of Yang Hui's triangle~(Pascal's triangle)
and admits the expression
\begin{equation}
    \T{\ell}{m} = \binom{m-2}{\ell}
\end{equation}
for $2 \leq \ell \leq m-2$ and $m \geq 4$.
Solving \eqref{eq:auxiliary-recursive-formula} for $\N{\ell}{m}$, we
obtain
\begin{equation*}
    \N{\ell}{m} = \T{\ell}{m} + \N{\ell-1}{m} = \sum_{j=2}^\ell \T{j}{m}
    + \N{1}{m} = \sum_{j=2}^\ell \binom{m-2}{j} + m-1
    = \sum_{j=0}^\ell \binom{m-2}{j}.
\end{equation*}
When $\ell = m-2$ and $m \geq 4$, the expression $\sum_{j=0}^{\ell} \binom{m-2}{j} =
2^{m-2}$ coincides with the above claim. This proves the theorem.
\end{proof}

\begin{corollary} \label{cor:oraclecapacity}
    Let $\F{\ell}{m}$ denote the capacity for the level $\ell$ recursive
    \emph{oracle circuit} on $m$ ancilla qubits. For various scenarios of
    $\ell$ and $m$, we have
    \begin{equation} \label{eq:corollary-recursive-formula}
        \F{\ell}{m} =
        \begin{cases}
            2^{m-1},  & \ell \geq m - 1 \\
            \sum_{j=0}^\ell \binom{m-1}{j}, & \mathrm{Otherwise}
        \end{cases}.
    \end{equation}
\end{corollary}

Corollary~\ref{cor:oraclecapacity} could be derived directly from
Theorem~\ref{thm:Ucapacity}. The capacities of $\U{\ell}{m+1}$ and the
level $\ell$ oracle circuit on $m$ ancilla qubits are the same. Hence we
could reorganize the capacity in Theorem~\ref{thm:Ucapacity} and lead to
Corollary~\ref{cor:oraclecapacity}.

\begin{table}[htb]
    \centering
    \begin{tabular}{c|cccccccccc}
        \toprule
        \multirow{2}{*}{$\ell$} & \multicolumn{10}{c}{$m$}\\
        \cmidrule(lr){2-11}
           & 1 & 2 & 3 & 4 & 5  & 6  & 7  & 8   & 9   & 10  \\
        \toprule
        1  & 1 & 2 & 3 & 4 & 5  & 6  & 7  & 8   & 9   & 10  \\
        2  & 1 & 2 & 4 & 7 & 11 & 16 & 22 & 29  & 37  & 46  \\
        3  & 1 & 2 & 4 & 8 & 15 & 26 & 42 & 64  & 93  & 130 \\
        4  & 1 & 2 & 4 & 8 & 16 & 31 & 57 & 99  & 163 & 256 \\
        5  & 1 & 2 & 4 & 8 & 16 & 32 & 63 & 120 & 219 & 382 \\
        6  & 1 & 2 & 4 & 8 & 16 & 32 & 64 & 127 & 247 & 466 \\
        7  & 1 & 2 & 4 & 8 & 16 & 32 & 64 & 128 & 255 & 502 \\
        8  & 1 & 2 & 4 & 8 & 16 & 32 & 64 & 128 & 256 & 511 \\
        9  & 1 & 2 & 4 & 8 & 16 & 32 & 64 & 128 & 256 & 512 \\
        10 & 1 & 2 & 4 & 8 & 16 & 32 & 64 & 128 & 256 & 512 \\
        \bottomrule
    \end{tabular}
    \caption{Maximum number of boolean equations in oracle circuit,
    $\F{\ell}{m}$, for $\ell$ being the recursive level and $m$ being the
    number of ancilla qubits.}
    \label{tab:maximum-allowed-equations}
\end{table}

Besides the capacitance of the oracle given a fixed number of ancilla
qubits, another concern for the oracle construction is the circuit depth.
We now calculate the total depth in terms of the number of
function-controlled gates. The circuit depth is perceived as a measure of
the time complexity for the quantum algorithm.

\begin{theorem} \label{thm:Udepth}
    Let $\K{\ell}{m}$ denote the number of function-controlled gates in
    $\U{\ell}{m}$. For various scenarios of $\ell$ and $m$, we have
    \begin{equation*}
        \K{\ell}{m} =
        \begin{cases}
            1,  & m = 1 \\
            2 \cdot 3^{m-2},  & \ell \geq m - 1, m \geq 2 \\
            \sum_{j=1}^\ell \binom{m-2}{j-1} 2^{j-1}
            + \sum_{j=0}^\ell \binom{m-2}{j} 2^{j}, & \mathrm{Otherwise}
        \end{cases}.
    \end{equation*}
\end{theorem}

\begin{corollary} \label{cor:oracledepth}
    Let $\G{\ell}{m}$ denote the number of function-controlled gates in a
    level $\ell$ oracle circuit on $m$ ancilla qubits. For various
    scenarios of $\ell$ and $m$, we have
    \begin{equation*}
        \G{\ell}{m} =
        \begin{cases}
            2 \cdot 3^{m-1},  & \ell \geq m - 1 \\
            \sum_{j=1}^\ell \binom{m-1}{j-1} 2^{j-1}
            + \sum_{j=0}^\ell \binom{m-1}{j} 2^{j}, & \mathrm{Otherwise}
        \end{cases}.
    \end{equation*}
\end{corollary}

In Theorem~\ref{thm:Ucapacity} and Corollary~\ref{cor:oraclecapacity}, we
give both expressions for the circuit capacities. Fixing the recursive
level $\ell$, both capacities are dominated by $\binom{m-2}{\ell}$ term as
$m$ goes large. Hence, given a BQE system with $R$ equations and recursive
level $\ell$, the number of required ancilla qubits scales as
$\bigO{R^{1/\ell}}$. In practice, we could calculate the capacities and
find the most proper number of ancilla qubits on classical computers
efficiently. Table~\ref{tab:maximum-allowed-equations} calculates the
capacities of oracles for various $\ell$ and $m$ for reference.

The asymptotic scaling of the capacity shows the power of our proposed
recursive oracle construction method. While in the era of noisy
intermediate-scale quantum~(NISQ), we are both limited by the number of
qubits and quantum circuit depth. Increasing $\ell$ increases the capacity
of the oracle. According to Corollary~\ref{cor:oracledepth}, the
circuit depth grows even faster than the capacities with an extra factor
$2^\ell$. Hence, given a BQE system, the proper $\ell$ and $m$ should be
carefully chosen based on the property and available resources of the
quantum computer.

\subsection{Oracle Compression}
\label{sec:oracle-compression}

The recursive oracle circuit above can be optimized to reduce the depth.
Recall that the function-controlled gates in the oracle circuit represent
product terms in the boolean equation. From the representation of the
boolean equation, we know that these gates should be commutable. From the
quantum gate perspective, these function-controlled gates consist of
(multi-)controlled \op{NOT} gates, controlling from some of the first $n$ qubits to an ancilla qubit
(including the \op{NOT} gates, which is the special case of controlling from none of the qubits).
All such controlled \op{NOT} gates are commutable. After commuting these
gates and other commutable gates, many gates could be eliminated without
affecting the circuit outcomes.

\begin{algorithm}
    \caption{Greedy oracle compression.}
    \label{alg:rearrange}
    \begin{algorithmic}[1]
        \Require a list of gates $\calG = [g_1, \dotsc, g_k]$.
        \Ensure a list of rearranged gates $\calR$.
        \State {$\calG = \mathrm{sort}(\calG)$}
        \For {each $g_i \in \calG$}
            \If {$g_i = g_{i+1}^\dagger$}
                \State {Eliminate $g_i$ and $g_{i+1}$ from $\calG$}
            \EndIf
        \EndFor
        \State {$\calR = []$}
        \While {$\calG \neq \emptyset$}
            \State {$\calQ = \emptyset$}
            \For {each $g_i \in \calG$}
                \If {qubits of $g_i$ are not in $\calQ$}
                    \State {Add qubits of $g_i$ to $\calQ$}
                    \State {Append $g_i$ to $\calR$}
                    \State {Eliminate $g_i$ from $\calG$}
                \EndIf
            \EndFor
        \EndWhile
    \end{algorithmic}
\end{algorithm}

An exact optimization of the recursive oracle circuit is a job--shop
problem with MPT, or $J;m|p_{ij} = 1;\mathrm{fix}_{j}|C_{\max}$
\footnote{The notation $J;m|p_{ij} = 1;\mathrm{fix}_{j}|C_{\max}$
describes a job--shop problem with $m$ dedicated machines. Each
multiprocessor task (job) requires a unit processing time simultaneously on
all the machines it specifies. The optimization task is to minimize the
makespan, i.e., the completion time of all tasks.} following the widely
adopted notation~\cite{SchedulingNotation}. Brucker et
al.~\cite{ShopScheduling} proves the strong NP-hardness of
$J;2|p_{ij}=1;\mathrm{fix}_j|C_{\max}$ which is a special case of $2$
qubits in this context.  This highlights the complexity of the exact
optimization which we consider infeasible in practice.  Instead of exact
optimization, we use a greedy algorithm to rearrange and reduce the
controlled \op{NOT} gates. There are two steps in our greedy algorithm:
1) eliminate gate pairs that are complex conjugates of each other; 2)
rearrange gates such that they can be maximumly parallelized. In the first
step, all commutable gates are sorted according to the qubit indices being acted on.
Then for all these gates, if they could cancel with
any other one, they should be neighboring to each other after sorting. Hence we go
through the sorted gate list, check whether neighboring gates are complex
conjugates of each other, and eliminate them. In the second step, we
search for maximumly parallelizable gates. We start a gate list with a
single gate. Then gradually add gates to the list that act on qubits
different from all other gates in the list. If no more gates can be added
to the list, we start a new gate list and repeat the process until all
gates are added to one of the gate lists. If the number of product terms
in each boolean equation is much more than the number of qubits, depth can
be reduced by a factor of $\frac{1}{m}$ for $m$ being the number of
ancillae. Figure~\ref{fig:rearrange-example} gives an example of a quantum
circuit before and after applying our greedy algorithm.

\begin{figure}[htb]
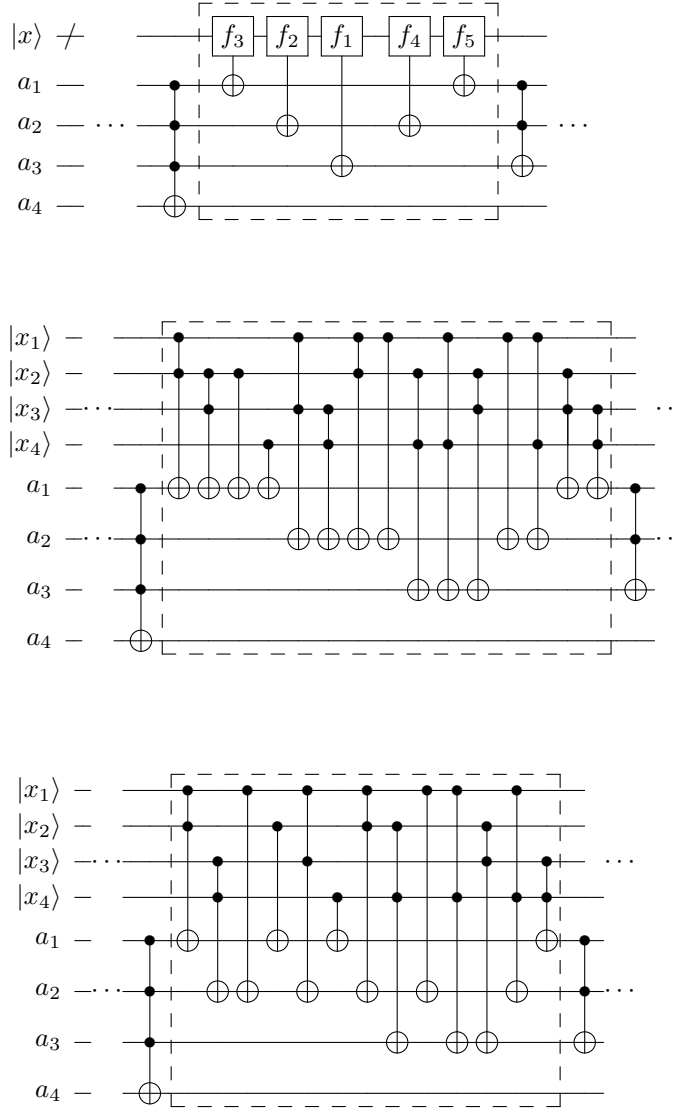

    \centerline{
        \xy
            \qcRearrangeOrigin
        \POS-(0,40)
            \qcRearrangeExpanded
        \POS-(0,60)
            \qcRearrangeAfter
        \endxy
    }
    \caption{A fragment of oracle for level $\ell = 2$ with $m = 4$
    ancilla qubits. The dashed boxes from top to bottom are
    function-controlled-gates, their expansion into \op{CNOT} and \op{MCX}
    gates, and gates after applying Algorithm~\ref{alg:rearrange}. The
    original circuit has a depth of $12$, whereas the optimized circuit has
    a depth of $8$.}
    \label{fig:rearrange-example}
\end{figure}

\section{Randomized Grover's Algorithm}
\label{sec:Algorithm}

In this section, we propose the random version of Grover's Algorithm with
the idea of randomly splitting boolean equations into groups. Then we
analyze the number of iterations required for different schemes of
splitting.


With an appropriate way of constructing the oracle, the vanilla Grover's
algorithm is complete after $K$ iterations. Recall $K$ for the highest
probability of obtaining a correct answer in one measurement is $K =
\mathrm{round} \left( \frac{\arccos{\sqrt{M/N}}}{\theta} \right)$.
However, $K$ could be very large in the case when $M \ll N$. As an
example, for $M = 1$ and $N = 2^{20}, 2^{25}, $ and $2^{30}$, $K$ are $568$,
$3217$, and $12198$, respectively. Combined with the complexity of each
Grover iteration, e.g., the circuit depth of the oracle, the problem
quickly becomes impractical.

\paragraph{Randomized Grover's algorithm.}
We propose a randomized Grover's algorithm that has different Grover
operators $G$ in each iteration. As discussed in
Section~\ref{sec:Construction}, the number of ancilla qubits limits the
number of equations it can solve. Splitting boolean equations into groups
and using different groups of equations in different Grover iteration
allows us to solve a boolean equation system with more equations in
given limited quantum resources.

Assume we have $R$ equations in total. Let all boolean equations be split
into $s$ groups, where groups are denoted as $\calR_1, \calR_2, \dotsc,
\calR_s$. Each group contains only part of $R$ equations and the union of
$\{\calR_i\}_{i=1}^s$ is all $R$ equations, i.e., $\bigcup_{i=1}^s \calR_i
= \{1, 2, \dotsc, R\}$. The number of equations in $\calR_s$ is denoted as
$R_s$. Note that groups can overlap with each other.

The randomized algorithm framework allows each Grover iteration to
choose one of $\calR_i$. Denote $s_i$ as the group index for the oracle in
the $i$-th iteration, i.e.,
\begin{equation}
    O_i: \quad \qubit{x} \mapsto
    \begin{cases}
        -\qubit{x}, & x \in \{ y \mid f_t(y) = 0, \forall t \in \calR_{s_i} \}\\
        \qubit{x},  & \text{otherwise}
    \end{cases}.
\end{equation}
Instead of iterating the same oracle, oracles are now different and are
labeled $O_1, \dotsc, O_K$ for a total of $K$ iterations. The
corresponding operators are labeled $G_1, \dotsc, G_K$, respectively. The circuit
before the measurement is
\begin{equation}
    G_K G_{K-1} \cdots G_1 H^{\otimes n} \qubit{0}
    = W O_K W O_{K-1} \cdots W O_1 \qubit{\psi},
\end{equation}
where $W$ is the diffusion operation.
Figure~\ref{fig:Randomized-Grover-framework} illustrates the procedure
with measurement.

\begin{figure}[htb]
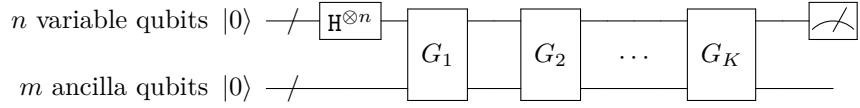

    \centerline{\qcGroverFrameworkRandomized}
    \caption{Randomized Grover's algorithm. Different operators $G_i$
    contain different oracles.}
    \label{fig:Randomized-Grover-framework}
\end{figure}

Many splitting strategies can be applied to the randomized framework. A
fixed splitting strategy is to split these equations in the first place,
$\cup_{i=1}^s \calR_i = \{1, \dotsc, R\}$. Then we iterate using $\calR_1,
\calR_2, \dots, \calR_s$ in a cyclic way until we reach the desired number
of iterations. Another randomized splitting strategy is to first estimate
the number of equations in each iteration, e.g., $q$ equations. Then, at
each iteration, we randomly chose $q$ distinct equations from the entire
boolean equation system. Oracle circuits are then constructed just for the
current iteration. It is worthwhile to point out that unless intended,
each equation should have an equal probability of being chosen in each
iteration. Otherwise, the algorithm will be biased to some equations.

The randomized Grover's algorithm is a natural random extension of the
vanilla Grover's algorithm. The major reason behind proposing the
randomized Grover's algorithm is to reduce the complexity of oracle
circuits, i.e., to reduce the depth and the number of ancilla qubits in the
oracle quantum circuits. However, the randomized idea does not always
work for all boolean equation systems with arbitrary grouping. Here we
give a tiny counterexample illustrating the failure of the randomized
Grover's algorithm.

Consider a 2-variable boolean equation system with two equations,
\begin{equation} \label{eq:simple-two-eq-bad-example}
    \begin{split}
        f_1(x_1, x_2) = & x_1 \oplus 1 = 0, \quad \text{and}\\
        f_2(x_1, x_2) = & x_2 \oplus 1 = 0.
    \end{split}
\end{equation}
Assume the randomized Grover's algorithm splits two boolean equations into
two groups and iterates between two groups in a cyclic way. More
precisely, the algorithm construct oracles with $f_1(x_1,x_2) = 0$ in
$G_1, G_3, \dotsc$, and $f_2(x_1,x_2) = 0$ in $G_2, G_4, \dotsc$. The
state vector before each iteration obeys
\begin{equation*}
    \frac12
    \begin{pmatrix}
        1\\
        1\\
        1\\
        1
    \end{pmatrix}
    \xrightarrow[]{G_1}
    \frac12
    \begin{pmatrix}
        -1\\
        -1\\
        1\\
        1
    \end{pmatrix}
    \xrightarrow[]{G_2}
    \frac12
    \begin{pmatrix}
        1\\
        -1\\
        -1\\
        1
    \end{pmatrix}
    \xrightarrow[]{G_3}
    \frac12
    \begin{pmatrix}
        -1\\
        1\\
        -1\\
        1
    \end{pmatrix}
    \xrightarrow[]{G_4}
    \frac12
    \begin{pmatrix}
        1\\
        1\\
        1\\
        1
    \end{pmatrix}
    \xrightarrow[]{G_5}
    \frac12
    \begin{pmatrix}
        -1\\
        -1\\
        1\\
        1
    \end{pmatrix}
    \xrightarrow[]{G_6}
    \cdots.
\end{equation*}
As shown in the above iterations, the state vector iterates cyclicly among
four vectors, and the amplitude of the solution entry does not change at
all. A quantum measurement after any number of iterations will give one of
all possible states with equal probability. Hence the randomized Grover's
algorithm fails in this case.

For the boolean equation system as in \eqref{eq:original-problem}, we
empirically find that we should not split equations into too many groups
with few equations. For a group with few equations, the solution set would
be much larger than the solution set of the boolean equation system. Such
a splitting would make randomized Grover's algorithm challenging to
succeed. A rigorous mathematical analysis for the success of our
randomized Grover's algorithm would be interesting for future work.

\paragraph{Estimating iteration numbers.}
Both Grover's algorithm and our randomized Grover's algorithm cannot
iterate forever and stop until they meet some convergence criteria.
The randomized Grover's algorithm,
similar to the vanilla Grover's algorithm, adopts a fixed number of
iterations, whose expression differs from \eqref{eq:GroverIterNum}.
The iteration number can be estimated directly by running the Grover operator
and computing the rotation angle of a single iteration, see \cite{QuantumCounting2023}.
In this part, we discuss the estimation of the iteration number for
randomized Grover's algorithm leveraging the expectation of randomness, which reduces the computational cost to the evolvement of a $2\times2$ matrix.

In quantum computing, almost all algorithms bear with randomness. The
measurement results are not deterministic. Desired solutions appear in the
measurement results with a relatively higher probability. Hence, quantum
algorithms, including Grover's algorithm and randomized Grover's
algorithm, have to be executed and measured many times until the desired
solutions appear once we can validate the solution efficiently. In
the case the solution can not be validated efficiently, we have to
repeat the execution and measurements many times and confirm the solution
through some form of majority voting. In solving nonlinear boolean
equations, we are able to validate the solution efficiently in polynomial
time complexity. Randomized Grover's algorithm (also Grover's algorithm)
further has a trade-off between the iteration number $K$ and the
measurement number $J$. The overall solving time for randomized Grover's
algorithm in addressing the boolean equation
system~\eqref{eq:original-problem} is considered as $J\cdot K$.
Hence we propose the following constraint optimization problem
to obtain the best iteration number $K$ and measurement number $J$,
\begin{equation} \label{eq:minimize-cost}
    \min_{P(J,K) > 1 - \varepsilon} J \cdot K,
\end{equation}
where $P(J, K)$ denotes the success probability of the boolean function
solver and $\varepsilon$ is a small failure probability of the algorithm.

In general, $P(J, K)$ increases monotonically with respect to both $J$ and
$K$ for $K$ smaller than the near-optimal number similar to
\eqref{eq:GroverIterNum}. However, an explicit expression for $P(J, K)$ is
unknown for randomized Grover's algorithm. Therefore, solving
\eqref{eq:minimize-cost} analytically and exactly becomes infeasible in
practice. In the following, we will first propose a simplified
probabilistic model for applying randomized Grover's algorithm to solve
boolean equation system, and then numerically estimate $J$ and $K$ as an
approximated solution of \eqref{eq:minimize-cost}.

Recall the $n$-variable boolean equation
system~\eqref{eq:original-problem} has $R$ boolean equations. Let the
number of solutions be $M$. Without loss of generality, we assume
solutions are the first $M$ states, i.e., $\qubit{0}, \qubit{1}, \dots,
\qubit{M-1}$. In randomized Grover's algorithm, each group of equations is
randomly selected from all $R$ equations, whose solutions contain that of
the original boolean equation system. Each oracle $O_i$ could be
represented by a $N \times N$ matrix,
\begin{equation} \label{eq:diagO}
    O_i = \mathrm{diag}
    \big(
        \underbrace{-1, \dotsc, -1}_{M},
        \underbrace{V_1^{(i)}, \dotsc, V_{N-M}^{(i)}}_{N-M}
    \big)
\end{equation}
where $N = 2^n$, and $V_j^{(i)}$ are $\pm 1$ depending on the selected
boolean equations in $O_i$. After $K$ iterations, the quantum state
vector of the randomized Grover's algorithm is
\begin{equation*}
    G_K \cdots G_1 \qubit{\psi} = W O_K \cdots W O_1\qubit{\psi},
\end{equation*}
and the probability of obtaining the $i$-th correct solution is
\begin{equation} \label{eq:pi}
    p_i = \left| \bra{i} G_K \cdots G_1 \qubit{\psi} \right|^2
\end{equation}
for $i = 0, 1, \dots, M-1$. Once $p_i$s are known, the success probability
of the boolean function solver $P(J, K)$ can be written down explicitly in
terms of $p_i$s, i.e.,
\begin{equation} \label{eq:minimize-constraint-pjk}
    P(J,K) = 1-(1-p)^J,\quad \text{where } p = \sum_{i=0}^{M-1} p_i.
\end{equation}
Note that this requires in the worst case, one has to examine all the $J$
outcomes of the process to find if any of the outcomes represents a correct
solution.  Checking whether an $x$ satisfies all $R$ equations is of
polynomial complexity. By \eqref{eq:minimize-constraint-pjk}, the
constraint of optimization \eqref{eq:minimize-cost} is $p > 1 -
\varepsilon^{1/J}$.

Now we make two assumptions to simplify the calculation of $p_i$s and
hence $P(J, K)$. First, we assume that all equation groups have exactly
$\widetilde{M}$ solutions for $\widetilde{M} > M$. Second, we assume that
$V_j^{(i)}$s are identically independently distributed random variables
for all $i$ and $j$. Those assumptions mean that the probability of every
$V_j^{(i)}$ admits
\begin{equation*}
    \Pr(V = 1) = \frac{N - \widetilde{M}}{N - M},
    \quad \Pr(V = -1) = \frac{\widetilde{M} - M}{N - M}.
\end{equation*}
With these two assumptions, we could calculate the expected value of
$G_i$, in the matrix format,
\begin{equation} \label{eq:expectedGi}
\begin{aligned}
    \bbE (G_i) = W \bbE (O_i) &= \frac{2}{N}
    \begin{pmatrix}
        1 \\
        \vdots \\
        1
    \end{pmatrix}
    \begin{pmatrix}
        -1 & \dotsb & -1 & \frac{N+M-2\widetilde{M}}{N-M} & \dotsb & \frac{N+M-2\widetilde{M}}{N-M}
    \end{pmatrix}\\
    & - \mathrm{diag}
    \left(
        -1, \dotsc, -1, \frac{N+M-2\widetilde{M}}{N-M}, \dotsc, \frac{N+M-2\widetilde{M}}{N-M}
    \right),
    \quad \mathrm{for}\; i = 1,\dotsc,K,
\end{aligned}
\end{equation}
and the expected quantum state vector is
\begin{equation} \label{eq:independentGi}
    \bbE(G_K \cdots G_1 \qubit{\psi})
    = (W \bbE(O_1))^K \qubit{\psi}.
\end{equation}
Notice that the right-hand side of \eqref{eq:expectedGi} is independent of
$i$ and $\bbE(W O_i) \qubit{x}$ can be represented only by two different
numbers. Thus, we could equivalently use a 2-dimensional matrix-vector multiplication
scheme to calculate the expected quantum state vector i.e.,
\begin{equation*}
    \qubit{\psi_2^{(i)}} = (G_{2\times 2})^i \qubit{\psi_2^{(0)}},
\end{equation*}
where $\qubit{\psi_2^{(i)}}$ denotes the $2$-dimensional expected quantum
state vector in $i$-th iteration, $G_{2\times 2} = W_{2\times 2}O_{2\times
2}$ denotes the $2$-dimensional Grover operator and
\begin{equation*}
    \qubit{\psi_2^{(0)}} = \frac{1}{\sqrt{N}}
    \begin{pmatrix}
        1 \\
        1
    \end{pmatrix}, \quad
    O_{2\times 2} = \begin{pmatrix}
        -1 & \\ & \frac{N+M-2\widetilde{M}}{N-M}
    \end{pmatrix}, \quad
    W_{2\times 2} = \frac{2}{N}
    \begin{pmatrix}
        1 \\
        1
    \end{pmatrix}
    \begin{array}{cc}
        \bigl(M, & N - M\bigr)\\
    \end{array}
    - I_{2\times 2}.
\end{equation*}
Therefore, the computed probability of obtaining each correct solution is
the square of the first element in $\qubit{\psi_2^{(K)}}$,
$p_i=\left|\braket{0|\psi_2^{(K)}}\right|^2$, for all $i=0,1,\dots, M-1$.
The probability of obtaining any correct solution is $p=M\cdot p_i = M
\cdot \left|\braket{0|\psi_2^{(K)}}\right|^2$. The optimization problem
\eqref{eq:minimize-cost} is thus
\begin{equation} \label{eq:minimize-cost-evaluated}
    \min_{M\cdot \left|\braket{0|\psi_2^{(K)}} \right|^2 > 1 -
    \varepsilon^{1/J}} J\cdot K.
\end{equation}

Given \eqref{eq:minimize-cost-evaluated}, we could conduct a brute force
search in the space of $J$ and $K$ for small-to-medium size boolean
equation systems and find the near-optimal $J$ and $K$, which is already
useful in NISQ.

Through the above analysis, the assumption that $\widetilde{M}$ remains
constant for all oracles is too strong to be true in practice. Further,
the number of solutions in oracles is also not known a priori. If an
efficient estimator or quantum algorithm in estimating $\widetilde{M}$ is
available, we could update our analysis above to be oracle dependent.
Otherwise, we have to adopt a rough estimation of $\widetilde{M}$ in
practice. For all numerical experiments, we estimate $\widetilde{M} = M
\cdot 2^{n-r}$, where $r$ is the number of equations used in each
iteration. From our testing, the success of the randomized Grover's
algorithm is related to the estimation of $\widetilde{M}$ but is not very
sensitive.

\section{Numerical Results} \label{sec:Numerical}

We focus on solving boolean equation systems with only boolean quadratic
equations (BQE) in this section. Each BQE is randomly generated in
advance, by uniformly selecting from all possible quadratic product terms
and linear terms. The number of total terms in BQE is set to a Poisson
random variable with the mean $\frac14 n(n+1)$ which equals half of all
available quadratic product terms. The solutions are calculated by brute
force for testing purposes. For some tests, we fixed the number of
solutions to a specific number or a range. If the number of solutions of
the randomly generated BQE system does not match the requirement, we add a
new BQE or drop an existing one until the requirement is satisfied. We
implement the randomized Grover's algorithm in Section~\ref{sec:Algorithm}
and run the algorithm in the state vector simulator using the state vector
simulation. Limited to the available resource, the maximum number of
variables we can simulate is set to $n=25$ in this section. In all
numerical results, programs are written using \texttt{Qiskit
0.39.3}~\cite{Qiskit} with \texttt{Python 3.8} and executed on a Linux
server with dual Intel Xeon Gold 6226R (2.90GHZ, $2\times16$ Cores,
$2\times32$ Threads) and 1 TB memory. Our code is available
at \url{https://github.com/j7168908jx/QuantumBooleanFunctionSolver}.

\subsection{Efficient Oracle Construction}

\paragraph{Capacitance and Recursive Level.} We explore the trade-off
between the recursion level $\ell$ and the number of ancilla qubits $m$ in
the oracle construction. A higher recursion level will reduce the number
of ancillae used in the circuit at the cost of a deeper circuit. As in the
previous sections, the number of gates will roughly double when we
increase the recursion level $\ell$ by one.

We run the oracle circuit generation with recursion levels set to $\ell =
1, 2, 3, 4$ on various randomly generated single-solution equations with
the number of variables between 10 and 25. The number of ancillae $m$ used
in the circuit is computed by counting the minimum ancillae required to
contain half of the total number of equations, i.e., $\F{\ell}{m} \geq
R/2$. For each pair of a variable number and an $\ell$, we adopt 15 randomly
generated BQEs to get rid of the randomness in the equation system
generation. The $95\%$ confidence intervals in our results are calculated
based on the 15 sample cases. To further explore the power of our
recursive oracle construction, we also include the number of ancillae when
the boolean equation systems have more than 25 variables. In these cases,
the underlying quantum circuits are not explicitly constructed.

\begin{figure}[htb]
    \centering
    \includegraphics[width=0.65\textwidth]{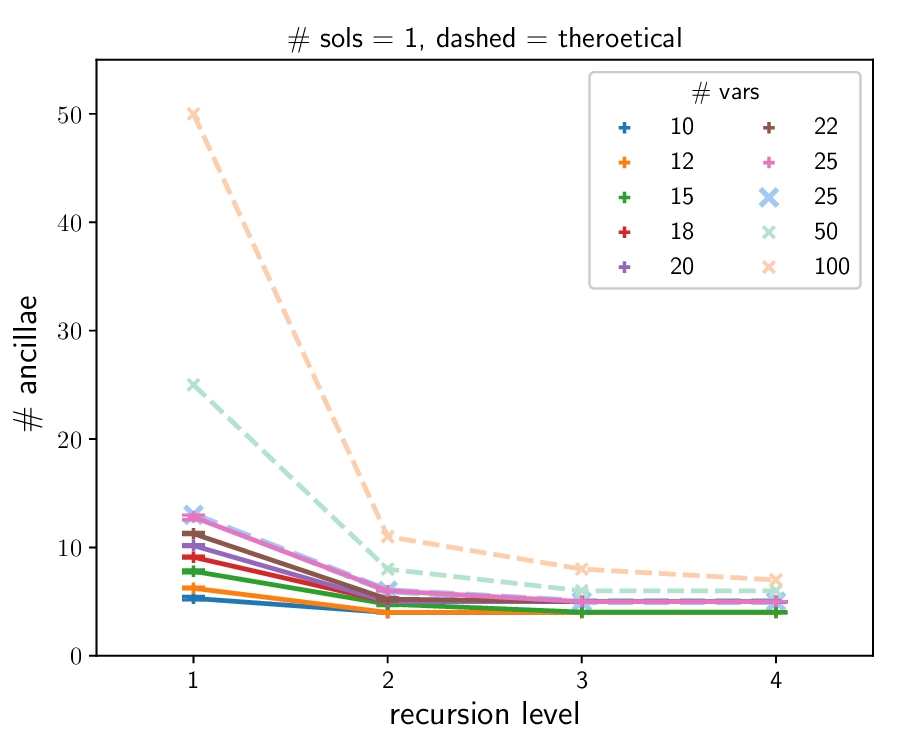}
    \caption{Relation between the recursive level and the required numbers
    of ancillae. The solid lines with `$+$' marks are experiments on
    randomly generated single-solution BQEs. The $95\%$ confidence
    intervals are plotted. The dashed lines with `$\times$' marks are
    theoretical results. Note that for $n=25$, both experimental and
    theoretical results are plotted and they overlap with each other.}
    \label{fig:relation-depth-and-level}
\end{figure}

Figure~\ref{fig:relation-depth-and-level} plots the required number of
ancillae for the oracles with different recursion levels. Solid lines
correspond to cases with actual quantum circuit constructions in practice
whereas dash lines correspond to our theoretical estimations. The solid
lines are plotted with confidence intervals. The variance of ancillae
usage comes from the variance of the number of equations, and is arguably small,
indicating that the scaling will not change much
for different BQEs. For $n = 25$, we plot both the experimental and
theoretical results. From the figure, these two results agree with each
other. As illustrated in Figure~\ref{fig:relation-depth-and-level},
increasing the recursion level will significantly reduce the number of
required ancillae, especially when the number of variables is large. For
BQEs with a small-to-medium number of variables, increasing the recursion
level from 1 to 2 drastically decreases the number of required ancillae.
While further increasing $\ell$ to 3 or above has little impact on
decreasing the number of required ancillae. Hence, for BQEs with a
small-to-medium number of variables, we suggest using $\ell = 2$, whereas
for BQEs with a large number of variables, we suggest using the
theoretical result to find a proper $\ell$ such that the required number
of ancillae is feasible on quantum computers.

\paragraph{Oracle compression.} We investigate the circuit depth before
and after our greedy oracle compression, Algorithm~\ref{alg:rearrange}. We
test 15 single-solution BQEs for each parameter setup, where the BQEs are generated as
described at the beginning of this section.

\begin{table}[htb]
    \centering
    \begin{tabular}{ccccccc}
        \toprule
        \multirow{2}{*}{$n$} & \multirow{2}{*}{$m$} & \multirow{2}{*}{Level
        $\ell$}
        & \multirow{2}{*}{\# Eq. per iter} & \multicolumn{3}{c}{Depth of
        one iteration} \\
        \cmidrule{5-7}
        &  &  &  & w/o compression & w/ compression & ratio \\
        \toprule
        15 & 7 & 1 & 7 & 839 & 149 & -82.24\% \\
        15 & 4 & 2 & 7 & 1459 & 802 & -44.98\% \\
        15 & 5 & 2 & 7 & 2373 & 1062 & -55.22\% \\
        20 & 4 & 2 & 7 & 2580 & 1382 & -46.44\% \\
        20 & 5 & 2 & 11 & 4198 & 1818 & -56.69\% \\
        20 & 5 & 3 & 15 & 8880 & 5333 & -39.94\% \\
        25 & 5 & 2 & 11 & 6520 & 2788 & -57.25\% \\
        25 & 6 & 2 & 16 & 9931 & 3527 & -64.49\% \\
        25 & 5 & 3 & 15 & 13904& 8258 & -40.61\% \\
        \bottomrule
    \end{tabular}
    \caption{Depth of a Grover iteration with and without greedy oracle
    compression. Various sizes of BQEs are tested and results are listed
    for comparison. For each circuit, the depth of one iteration is the
    average circuit depth from the whole circuit.
    The depth of a circuit is counted based on the predefined simple gates in qiskit.
    The displayed figures in the table are further averaged among 15 distinct BQEs (15 distinct
    circuits).}
    \label{tab:effect-of-rearrange}
\end{table}

Table~\ref{tab:effect-of-rearrange} presents the averaged circuit depth
with and without our greedy oracle compression. First of all, for all
cases in Table~\ref{tab:effect-of-rearrange}, our greedy oracle
compression reduces the circuit depth by at least $\sim 40\%$. In most
cases, the circuit depths are reduced by half. From the 6th and last row
of Table~\ref{tab:effect-of-rearrange}, where the recursion level is
relatively high $\ell = 3$, we find that their compression ratios are low.
As the circuit contains more and more \op{MCX} gates on the same set of
qubits and ancilla qubits, the number of quantum gates that could be commuted and
parallelized are reduced. Hence the compression ratio is low in these
cases. Therefore, empirically, we observe that our greedy oracle
compression technique is more powerful when the recursion level $\ell$ is
relatively small. For BQEs with a small-to-medium number of variables, we
again suggest using $\ell = 2$.

\subsection{Randomized Grover's Algorithm}

\paragraph{Total circuit depth.}
Splitting equations into various number of groups in randomized Grover's
algorithm will affect both the total depth of the circuit and the number
of ancilla qubits required. In this part, we will explore the trade-off
between the number of groups and the total depth.


In this experiment, we compare the depth with different splitting
strategies for BQEs with $15$ and $20$ variables. We randomly generate a
sequence of BQEs with the solution number smaller than 90 for both $n =
15, 20$. We group these BQEs into five based on their solution numbers,
namely, single-solution, 2-4 solutions, 5-9 solutions, 10-19 solutions,
and 20-89 solutions. Each group has 40-100 BQEs. For each BQE and a split
factor, the quantum circuit is constructed with recursion level $\ell = 2$
and Grover iteration number achieving a nominal success rate of $99.9\%$
in 1024 shots. The split factor is defined as the total number of
equations over the number of equations chosen per Grover iteration.
A higher split factor means each iteration contains fewer
equations. The nominal success rate is the one estimated by our computing
model in Section~\ref{sec:Algorithm}, which is different from the actual
success rate.


\begin{figure}[htb]
    \centering
    \begin{subfigure}[b]{0.48\linewidth}
      \centering
      \includegraphics[width=\linewidth]{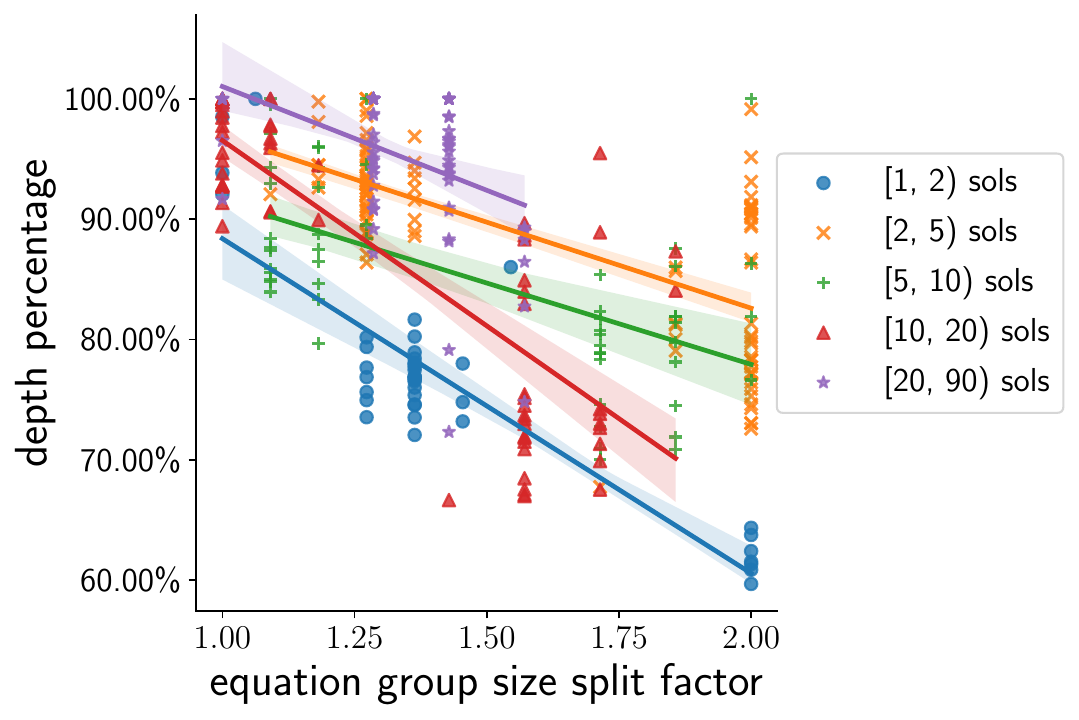}
    \end{subfigure}
    ~
    \begin{subfigure}[b]{0.48\linewidth}
      \centering
      \includegraphics[width=\linewidth]{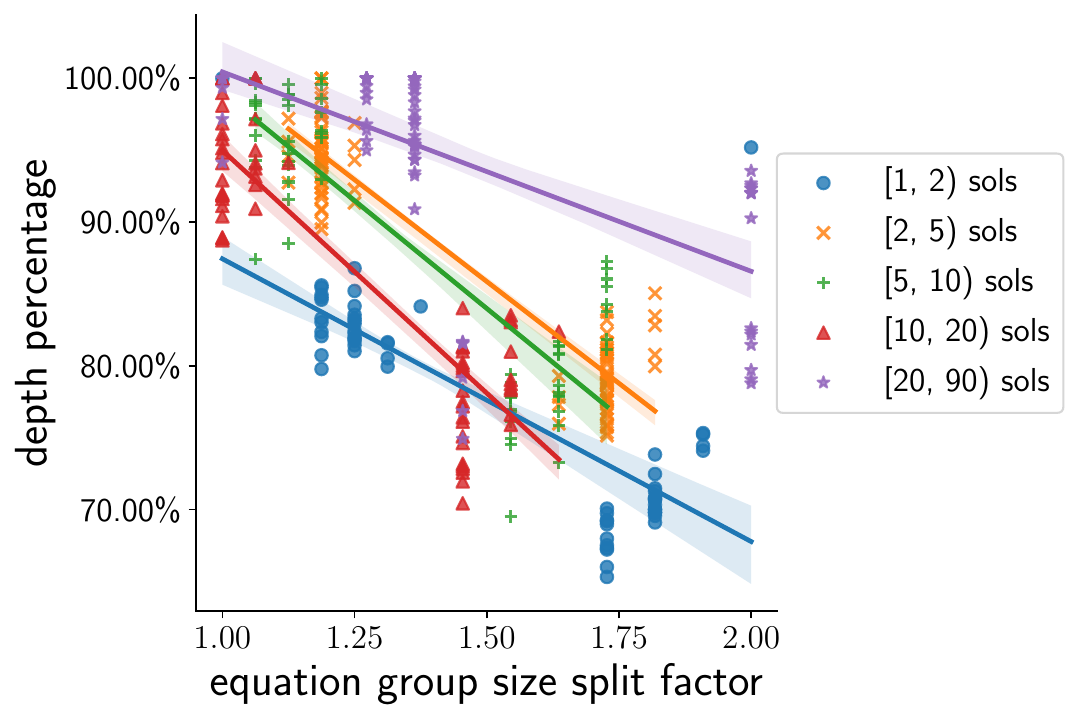}
    \end{subfigure}
    \caption{Relation between splitting and the total depth. The left figure plots 15 variables BQEs and the right plots 20 variable BQEs. For both figures, the $x$-axis is the split factor and the $y$-axis shows the relative depth in each
    group, so the longest circuit in each group is normalized as $100\%$,
    and others are the relative depth to the longest one. The regression
    line for each group is plotted with the $95\%$ confidence interval. }
    \label{fig:relation-split-and-iteration}
\end{figure}

Figure~\ref{fig:relation-split-and-iteration} illustrates the relative
depth of the circuit required to obtain a $99.9\%$ success rate in 1024
shots. The relative depth is defined as the depth of the circuit over the
maximum depth within the group. Each BQE with a split factor contributes
to a point in Figure~\ref{fig:relation-split-and-iteration}. For a better
understanding of the data, we conduct a linear regression for each BQE
group and plot the regression lines together with their $95\%$ confidence
interval. The regression line of the 15-variable 20-to-90-solution group
is shorter since the number of equations in each iteration is already
small.

For both figures in Figure~\ref{fig:relation-split-and-iteration}, all
slopes of the regression lines are negative. Hence, the total depth of the
quantum circuit is reduced as the proportion of equations used in each
iteration decreases. Here, we only explore the split factor between $1.0$
and $2.0$. For a small-to-medium BQE, a large split factor is not helpful.
In general, the circuit depth will first decrease and then increase as the
split factor increases. When the split factor is relatively close to 1,
the splitting strategy effectively reduces the total circuit depth and our
computing model in Section~\ref{sec:Algorithm} is trustable. However, when
the split factor is relatively large, each split group has too few
equations and yields a vast solution space. In this scenario, the number
of Grover iterations must be large and the total circuit depth is also
large. Empirically, we find that a split factor around 2 is efficient for
BQEs with a small-to-medium number of variables.

\paragraph{Success rate.}

On some very small-scale problems, the failure rate of randomized Grover's algorithm
will be abnormally high, due to the inaccurate estimate of $\widetilde M$.
For these small-scale problems, the original solution space is large enough to converge in a
very small amount of iterations, making saving resources less important.
The variance in estimating the iterations dominates the actual value,
leading to a very inaccurate outcome.
We aggregate 15 randomly generated circuits from 10 to 18 variables and summarize the
observed success rate for different iteration numbers in Figure~\ref{figure:failure}, where each experiment sets the nominal success rate
to $80.0\%$.
Firstly, the success rate for most groups has reached the nominal success rate, especially for all groups with 256 shots. The average success probability is $92.4\%$, which is higher than the nominal success rate.
Secondly, due to the limit of accurately estimating the number of solutions in each iteration of Grover's algorithm, we observed that the success rate is lower than the nominal success rate for some cases. Given the number of shots, this discrepancy becomes more prominent as the number of variables increases from $10$ to $18$.
This is because, for a fixed number of shots, a higher number of Grover iterations is required with sparser solutions, which increases the likelihood of overshooting or undershooting the desired result. It is better to increase the number of shots, rather than the number of iterations, to obtain a higher success rate.
Finally, note that some circuits cannot reach a nominal success rate of $80\%$ if the number of shots is set to less than $16$. This is caused by the success rate of a single shot cannot be arbitrary large, and the fail rate of a single shot to the power of the number of shots will not fall below $20\%$.
In this case, one has to increase the number of shots to obtain a higher success rate.

\begin{figure}[htb]
    \centering
    \includegraphics[width=0.9\textwidth,trim={0 0 0.5cm 0},clip]{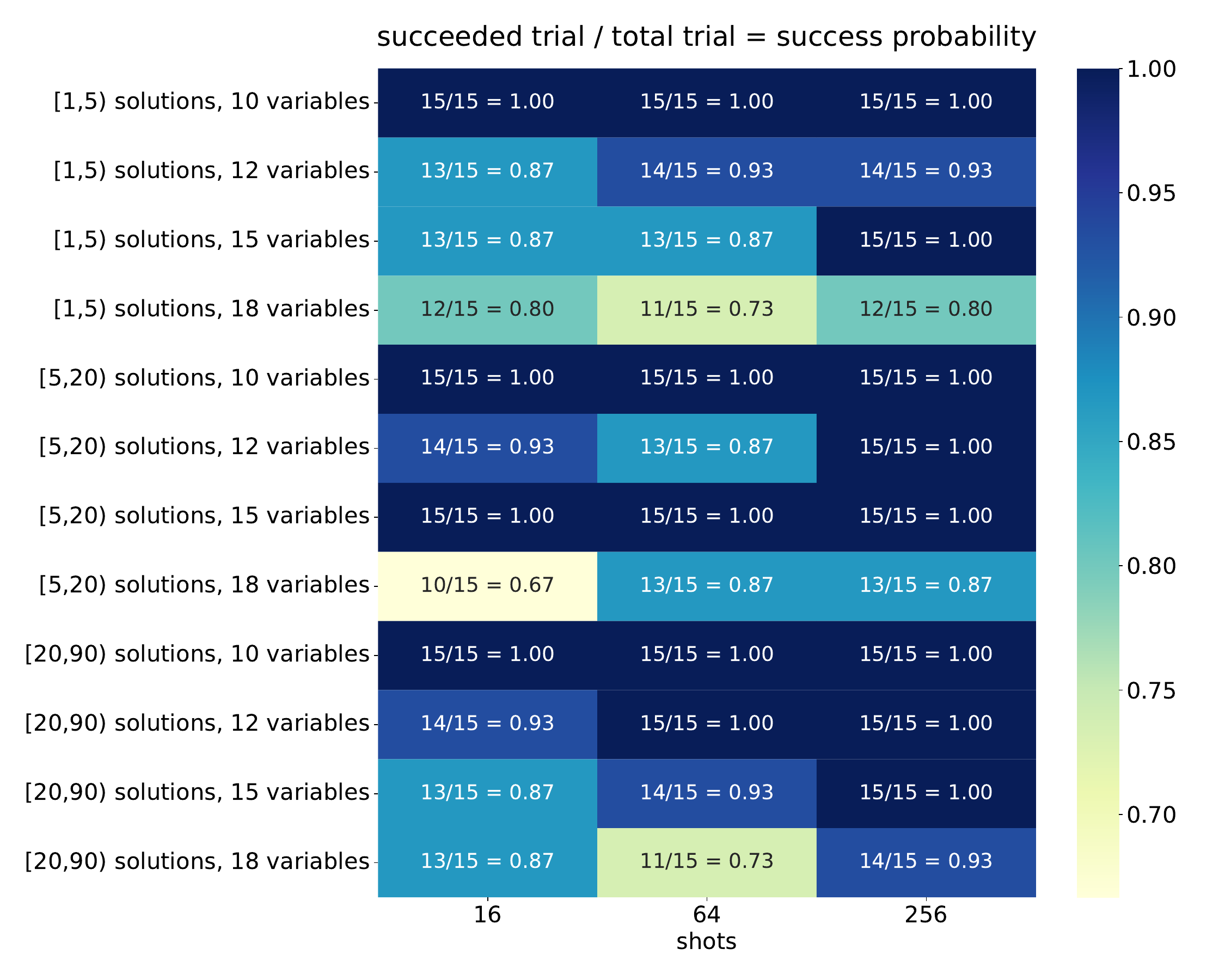}
    \caption{The success rate of each circuit running several numbers of shots.
    The $x$-axis shows the number of shots we ran for each circuit. The $y$-axis contains the number of variables and the number of solutions for each circuit.
    In each cell of the heatmap, we annotate the number of succeeded circuits, the number of total circuits we ran, and the success rate.  The color bar plots the value of the success rate.
    Due to the constraint of computational resources, we can only include test cases with no more than 18 variables.}
    \label{figure:failure}
\end{figure}

\section{Conclusion and Future Work} \label{sec:Conclusion}

We proposed three novel techniques to efficiently solve nonlinear boolean
equations on quantum computers under Grover's algorithm framework. Three
techniques are W-cycle oracle construction, oracle compression, and
randomized Grover's algorithm. For the first technique, W-cycle oracle
construction improves the capacity in encoding the boolean equations into
the quantum circuit given a fixed number of ancilla qubits. W-cycle oracle
construction introduces a recursive circuit construction idea to maximally
reuse ancilla qubits to keep solution information. Given $m$ ancilla
qubits, vanilla quantum circuit construction could encode $m$ boolean
equations, whereas our W-cycle oracle construction could at most encode
$2^m$ boolean equations at the cost of deeper circuits. The W-cycle oracle
construction also introduces a way to conduct the trade-off between the
number of ancilla qubits and the circuit depth. The flexible trade-off is
very important in the NISQ era. The second technique, oracle compression,
adopts a greedy strategy to swap commutable quantum gates. Oracle
compression eliminates redundant quantum gate pairs and rearranges the
quantum gates in a parallelizable way. Numerical experiments show that the
oracle compression technique leads the quantum circuit depth saving by a
factor between $40\%$ to $80\%$. The saving rate, in general, is larger
when the recursive level $\ell$ is small in the W-cycle oracle
construction. The third technique, randomized Grover's algorithm, uses
random combinations of boolean equations in each iteration to construct
the oracle and reduces the oracle circuit depth. In each iteration, the
algorithm chooses parts of the boolean equation system. The solution set
of the chosen equations contains that of the original system. Hence,
through the randomized Grover iteration, amplitudes of the original
solutions are always amplified. While the amplitudes of those virtual
solutions (solutions of chosen equations but not the original system) are
amplified in some iterations and damped in other iterations. However, the
algorithm is not guaranteed to converge. To improve the convergence, we
empirically make the number of equations per iteration relatively large.
An estimation of the number of iterations for randomized Grover's
algorithm is also proposed. The estimation analysis is carried out under
the assumption that the number of solutions in each iteration is a constant.
Numerically, we find that our randomized Grover's algorithm is efficient
and the estimation of iteration number is useful.

There are a few interesting future directions. The most interesting one
would be extending the randomized Grover's algorithm to a wider range of
applications. The idea behind randomized Grover's algorithm is to randomly
relax the constraints on the solutions so that the quantum circuit for the
oracle could be significantly simplified. Such an idea would be useful not
only in NISQ but also in beyond--NISQ era. The rigorous analysis for
randomized Grover's algorithm would be another interesting future
direction. A careful and rigorous analysis would hint at the choices for
selecting boolean equations in each iteration and also lead to a more
accurate estimation of the iteration number. Other interesting future
directions include extending the W-cycle circuit construction idea to
other applications and exploring other optimizing techniques like multiplying
some equations into one.

\section{Acknowledgments} \label{sec:Acknowledgments}
We thank
National Natural Science Foundation of China (NSFC),
National Key R\&D Program of China,
Science and Technology Commission of Shanghai Municipality (STCSM) and
Shanghai Institute for Mathematics and Interdisciplinary Sciences (SIMIS)
for their financial support.
This research was funded by
NSFC under grant 12271109 and 71991471,
National Key R\&D Program of China under grant 2020YFA0711902,
STCSM under grant 22TQ017 and 24DP2600100, and
SIMIS under grant number SIMIS-ID-2024-(CN).
The authors are grateful for the resources and facilities provided by
NSFC, National Key R\&D Program of China, STCSM and SIMIS,
which were essential for the completion of this work.
\bibliographystyle{alphaurl}
\bibliography{main.bib}

\end{document}